\documentclass[journal,draftcls,onecolumn,12pt,twoside]{IEEEtran}
\usepackage{graphicx,multirow}
\usepackage{epsfig}
\usepackage{algorithm}
\usepackage{algorithmic}
\usepackage{cite}
\usepackage{amssymb}
\usepackage{amsthm}
\usepackage{amsmath}
\usepackage{color}
\usepackage{bbm}
\usepackage{cite}
\usepackage{epstopdf}
\usepackage[table,xcdraw]{xcolor}
\newtheorem{theorem}{\textbf{\text{Theorem}}}
\newtheorem{lemma}{\textbf{\text{Lemma}}}

\newtheorem{corollary}{Corollary}
\newtheorem{approximation}{Approximation}

\newtheorem{remark}{Remark}
\usepackage{subfigure}
\usepackage{collcell, datatool}
\usepackage{pifont}

\definecolor{lightgray}{gray}{0.9}
\usepackage{algorithm}
\usepackage{algorithmic}
\usepackage[algo2e]{algorithm2e}
\usepackage{mathrsfs}

\begin{document}

\title{Spatiotemporal Model for Uplink IoT Traffic: Scheduling \& Random Access Paradox \\
}
\author{
\IEEEauthorblockN{\large  Mohammad Gharbieh, Hesham ElSawy, Hong-Chuan Yang, \\ Ahmed Bader, and Mohamed-Slim Alouini\\
\thanks{ M. Gharbieh and  H.-C. Yang are with the Department of Electrical and Computer Engineering, University of Victoria, Victoria, BC V8P 5C2, Canada; e-mail: \{mohammadgharbieh, hy\}@uvic.ca. { H. ElSawy and M.-S. Alouini are with the Computer, Electrical, and Mathematical Sciences and Engineering (CEMSE) Division, King Abdullah University of Science and Technology (KAUST), Thuwal 23955, Saudi Arabia; e-mails: \{hesham.elsawy, ahmed.bader, slim.alouini\}@kaust.edu.sa. {A. Bader is with Insyab Wireless Limited, 1961 Dubai - UAE; e-mail: ahmed@insyab.com.}}}
}}
\maketitle
\thispagestyle{plain}
\pagestyle{plain}

\begin{abstract}
The Internet-of-things (IoT) is the paradigm where anything will be connected. There are two main approaches to handle the surge in the uplink (UL) traffic the IoT is expected to generate, namely, Scheduled UL (SC-UL) and  random access uplink (RA-UL) transmissions. SC-UL is perceived as a viable tool to control Quality-of-Service (QoS) levels while entailing some overhead in the scheduling request prior to any UL transmission. On the other hand, RA-UL is a simple single-phase transmission strategy. While this obviously eliminates scheduling overheads, very little is known about how scalable RA-UL is. At this critical junction, there is a dire need to analyze the scalability of these two paradigms. To that end, this paper develops a spatiotemporal mathematical framework to analyze and assess the performance of SC-UL and RA-UL. The developed paradigm jointly utilizes stochastic geometry and queueing theory. Based on such a framework, we show that the answer to the {\em ``scheduling vs. random access paradox''} actually depends on the operational scenario. Particularly, RA-UL scheme offers low access delays but suffers from limited scalability, i.e., cannot support a large number of IoT devices. On the other hand, SC-UL transmission is better suited for higher device intensities and traffic rates.

\begin{IEEEkeywords}
Internet of Things (IoT), Uplink transmission, Grant-based access, Grant-free access, Stochastic geometry, Queueing theory, Stability analysis, Interacting queues.
\end{IEEEkeywords}

\end{abstract}

\section{Introduction}
\IEEEPARstart The Internet-of-things (IoT) is penetrating to different vertical sectors (e.g., smart cities, public safety, healthcare, autonomous deriving, etc.) which will bring billions of new devices to the already congested wireless spectrum. A recent report from ABI Research predicts that 75\% of the growth in wireless connections between today and the end of the decade will come from non-hub devices, e.g., sensor nodes \cite{ABI.IoT}. Scalability is one of the major challenges for the next frontier of IoT networks \cite{what.will.5g.be}. Such scalability is essential to accommodate the surging machine-type communications within the proliferating IoT applications. Accordingly, new wireless technologies should be developed to serve such unprecedented traffic, which is essentially a blend of human-type and machine-type communications. The challenge is more acute in the uplink (UL) direction since most of the IoT applications are UL-centric by nature~\cite{Bader_magazine}. This underlines the utter need {for increasing the UL data transmission capacities and the efficiency of medium access schemes~\cite{Bader_magazine, Goodbye, MD_magazine}}. It is a matter of fact that the contention over scarce resources for medium access escalates substantially as the number of devices and traffic intensity grow. This can cause excessive access delays leading to frequent packet drops \cite{RACH2}.

In the realm of Release-13, Release-14, and beyond, the 3GPP provisions few IoT-specific technologies (e.g., EC-GSM-IoT, LTE-eMTC, and NB-IoT) in order to accommodate IoT traffic~\cite{Release14, Primer_NBIOT}.  The 3GPP solutions adopt a scheduled UL (SC-UL) approach for the sake of interference management and guaranteed QoS provisioning. The SC-UL involves a random access scheduling request (RA-SR) prior to resource allocations. This is because the BS should be first notified upon data generation at the device buffer. The RA-SRs are not supervised by the BS and are subject to intracell and intercell interference. The BS then provides exclusive access transmission (EA-Tx) resource blocks for successful RA-SRs, and hence, EA-Tx transmission does not experiences intracell interference. Note that a single successful RA-SR can {secure} EA-Tx over several subsequent frames.  Given the sporadic traffic of the IoT devices, such two-phase SC-UL (i.e., RA-SR then EA-Tx) scheme may impose an unnecessary delay. The RA-SR overhead  becomes significant for shorter EA-Tx transmission periods. Furthermore, the multi-phase handshaking processes (i.e., scheduling request, scheduling response, resource allocation, and EA-Tx transmissions) impose longer wakeup time and data processing for the IoT devices, which shortens the battery lifetime of the IoT devices~\cite{LORA-NBIOT}. To alleviate such delay, signaling, and power consumption overhead, several low-power-wide-area (LPWA) networks (e.g., LoRa and Sigfox) adopt a single phase random access UL (RA-UL) data transmission scheme~\cite{LORA-NBIOT, Goodbye}. Each IoT device persistently transmits the data from its buffer over a randomly selected  resource block. Relying on the sporadic data pattern at each of the IoT devices, prompt RA-UL data transmission is expected to help devices to flush their buffer soon after data generation. Consequently, RA-UL scheme has the potential to reduce transmission delay, however, due to the probability that more than one IoT device may choose the same resource block, the data transmission is exposed to intracell interference. As such, the SC-UL scheme {experience} intracell interference in the RA-SR phase only whilst the  RA-UL scheme {may suffer}  intracell interference in every data transmission. Since each transmission scheme has its own intuitive merits, { there is }a pressing need for a mathematical framework that characterizes the tradeoff between both transmission schemes and identifies the effective operational scenario of each scheme.

\subsection{Related Work}
\vspace{-2mm}
In this context, stochastic geometry serves as the baseline model for massive and interference limited networks~(see \cite{WinPinShe:J09, Haenggi_Survey, survey_h, hesham_tutorial} and the references therein). However, stand-alone stochastic geometry analysis fails to account for the temporal aspects of the IoT such as traffic intensity per device and data accumulation in buffers. The temporal aspects within scheduling problems are usually captured via interacting queueing models~\cite{Neely2, Tassiulas1, Modiano1}, however, the work in \cite{Neely2, Tassiulas1, Modiano1} does not account for the spatial aspects (e.g., node density and mutual interference) that govern the interactions between the queues. Recently, spatiotemporal models that integrate stochastic geometry and queueing theory are proposed to jointly account for per-device traffic intensity, spatial device density, medium access control (MAC) scheme, devices' buffer states, and the mutual interference between devices\textcolor{red}{\cite{8198810,7486114,Chisci,Nardelli_Stability, Zhuang,8335767, Tony,  833,8439089, Gharbieh_tcom }}. Thus, {the IoT network can be abstracted by} a network of spatially interacting queues, where the interactions among the devices are governed by a signal-to-interference-plus-noise-ratio (SINR) capture model. However, \cite{8198810,7486114,Chisci,Nardelli_Stability, Zhuang,8335767} focus on ad hoc networks. Specifically, \cite{7486114,8198810,Chisci} discus the stability of ad hoc networks. The work in \cite{Nardelli_Stability} investigates throughput optimization in ad hoc networks with spatial randomness. \cite{Zhuang} assess the performance of ad hoc networks with unsaturated traffic in terms of the outage probability and the average packet delay. The authors in \cite{8335767} analyze the tradeoff between the delay and the security performance in ad hoc networks. The work in \cite{Tony} addresses downlink scheduling, which does not involve an RA-SR phase because the BS encloses the data queue and is responsible for scheduling. \textcolor{red} {The work in \cite{833,8439089}  analyzes the Random Access CHannel (RACH) performance in cellular-based IoT networks which does not involve EA-Tx phase for data transmission.} The UL scenario where the data queue is at the device side is considered in~\cite{Gharbieh_tcom}. However, only RA-UL data transmission with power ramping and transmission deferral are considered.  To the best of the authors' knowledge, the 3GPP based SC-UL, with joint RA-SR then EA-Tx phases, for IoT UL traffic is never addressed in the literature. Furthermore, there are no existing studies that assess and compare the 3GPP based SC-UL (i.e., grant-based access) and the LPWA based RA-UL (i.e., grant-free access) schemes for IoT networks. 

\vspace{-2mm}
\subsection{Contributions}
\vspace{-2mm}
This paper presents a novel spatiotemporal model for IoT UL communications to characterize the SC-UL and RA-UL transmission schemes.\footnote{This work is presented in part in~\cite{Gharbieh_ICC}.} From the spatial perspective, the BSs and devices are modeled via two independent Poisson point processes (PPPs). Besides its simplicity and tractability, the PPP model is validated via several experimental and theoretical studies~\cite{WinPinShe:J09, survey_h, hesham_tutorial}. From the temporal perspective, each communication link is modeled via a discrete time Markov chain (DTMC) to track packets generation and departure from devices' buffers. Consequently, the overall developed spatiotemporal framework models the IoT devices as a network of spatially interacting DTMCs, where the interactions are governed by the SINR capture model. To this end, the developed mathematical framework is used to characterize the scalability of SC-UL and RA-UL transmission schemes. The scalability of the network is investigated via the spatiotemporal traffic demand along with the notion of queueing stability. Within this context, a stable buffer (i.e., queue) is the one that has packet departure probability greater than the packet arrival probability \cite{loyens}. Otherwise, the number of packets in the devices' buffers would grow indefinitely driving these buffers (i.e., the network) to a state of  { ``instability"}. As such, an IoT network can scale in terms of devices and/or traffic rates as long as it still falls within the stability region. In this paper, scalability is characterized by the Pareto-frontier of all pairs of devices density and per-device traffic intensity that keeps the devices' buffers stable. Consequently, the terms scalability and stability are hereafter interchangeable. The contributions of the paper can be summarized as follows:

 \begin{itemize}
 \item SC-UL transmission with the joint RA-SR and EA-Tx phases of the 3GPP is modeled via a tandem queueing approach. An SINR capture model is adopted, where the SINR has to exceed a given threshold for packet departure. Since several devices served by a given BS may simultaneously select the same resource  for RA-SR, only the device with the highest SINR succeeds if its SINR exceeds the RA-SR SINR threshold. The EA-Tx phase enforces a single device per channel per BS, and hence, the UL SINR threshold is the only constraint for transmission success.

  \item The RA-UL scheme is modeled via a baseline DTMC. An SINR capture model is adopted, where the SINR has to exceed a given threshold for packet departure. Since several devices served by a given BS may simultaneously utilize the same resource for RA-UL, only the device with the highest SINR succeeds if its SINR exceeds the UL SINR threshold. It is important to note that the developed RA-UL model in this paper is different from \cite{Gharbieh_tcom}, which constrains packet departure with the UL SINR threshold only. Hence, the RA-UL model  in \cite{Gharbieh_tcom} is more optimistic as all intracell interfering devices satisfying the SINR threshold simultaneously deliver their packets to the serving BS.

     \item The SC-UL and RA-UL techniques are compared in terms of transmission  success probability, delay, average queue size, and scalability. The effective operational scenario of each transmission scheme is identified. For instance, RA-UL transmission is effective for lower device intensity with high traffic demand per each device. However, when the devices intensity grows, intracell interference becomes overwhelming and  scheduling is necessary to maintain stability.
 \end{itemize}

\vspace{-2mm}
\subsection{Notation \& Organization }
\vspace{-2mm}
Throughout the paper, we use the math italic font for scalars, e.g., $x$. Vectors are denoted by lowercase  math bold font, e.g., $\mathbf{x}$,
while matrices  are denoted by uppercase  math bold font, e.g.,  $\mathbf{X}$. We use the calligraphic font, e.g.,  $\mathcal{X}$ to represent a random variable (RV) while the math typewriter font, e.g., $\mathtt{x}$ is used to represent its instantiation. Moreover, $\mathbb{E}_\mathcal{X} \{\cdot\}$ and $\mathscr{L}_\mathcal{X}$ denote, respectively, the expectation and  the Laplace Transform (LT) of the PDF of the random variable $\mathcal{X}$. We use $\mathbb{P} \{\cdot\}  $ to denote the probability and the over-bar, e.g, $(\bar{\cdot})=(1-\cdot)$ to denote the probabilistic complement operator. $\mathbbm{1}_{\{ \cdot \}}$ is the indicator function which has value of one if the statement $\{ \cdot \}$ is true and zero otherwise. $\Gamma(\cdot)$ indicates the Gamma function and ${}_2 F_1(.)$ is the Gaussian hypergeometric function. Finally, $(\cdot)_{[m]}$ denotes the value at the $m^{th}$ iteration. 

The rest of the paper is organized as follows. Section \ref{sic:System_Model} presents the system model, the approximations, and methodology of analysis. Section \ref{sic:Performance_Analysis} presents the spatial, temporal, and spatiotemporal analysis of the depicted IoT network for both the SC-UL and RA-UL  schemes. Section \ref{sec:Results} presents and discusses some numerical results. Finally, Section \ref{sec:Conclusions} summaries and concludes the paper.

\section{System Model and Approximations}\label{sic:System_Model}

\subsection{Spatial \& Physical Layer Parameters}
We consider a single-tier network where the BSs are spatially distributed in $\mathbb{R}^2$ according to a homogeneous Poisson point process (PPP) $\boldsymbol{\Psi}$ with intensity $\lambda$. The devices are spatially distributed in $\mathbb{R}^2$ via an independent PPP  $\boldsymbol{\Phi}$ with intensity $\mu$. Without loss of generality, each device is assumed to associate to its nearest BS. Hence, the average number of devices associated to each BS is denoted as $\alpha = \frac{\mu}{\lambda} $.  A power-law path-loss model is considered where the signal power decays at a rate $r^{-\eta}$ with the propagation distance $r$, where $\eta>2$ is the path-loss exponent. In addition to the path-loss attenuation, Rayleigh flat fading  is assumed, in which all the channel power gains ($h$) are exponentially distributed with unity power gain. All channel gains are assumed to be independent of each other, independent of the spatial locations, and are identically distributed (i.i.d). It is also assumed that the channel  power gains are independent across different time slots for all the devices. {Fig.\ref{Network} shows a netwrok realization for $\alpha=16$ within an area of $4$ km$^2$.}

\begin{figure}[t!]
	\begin{center}
		\vspace{-10mm}
		\includegraphics[width=5 in ]{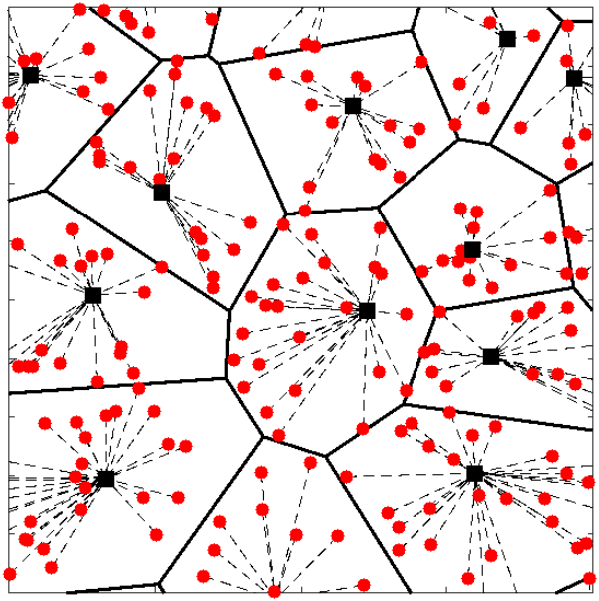}
	\end{center}
	\vspace{-15mm}
	\caption{  Network realization for the system model for Devices-to-BS ratio $\alpha=16$ within an area of $4$ km$^2$. The {\rm BS}s are denoted by black squares and the devices are denoted by the red circles. The Voronoi cells of the {\rm BS}s are denoted by the solid black lines while the black dashed lines denote the associations of the devices to the BSs.}\label{Network}
\end{figure}

All UL transmissions utilize full path-loss inversion power control with threshold $\rho$. That is, each device controls its transmit power such that the average signal power received at its serving BS is equal to a predefined value $\rho$. It is assumed that the BSs are dense enough such that each of the devices can invert its path loss towards the closest BS almost surely, and hence the maximum transmit power of the IoT devices is not a binding constraint for packet transmission. Extension to fractional power control and/or adding a maximum power constraint can be done by following the methodologies in \cite{uplink2_jeff} and \cite{uplink_alamouri}.
\begin{figure}[t!]
	\begin{center}
		\includegraphics[width=3.25 in ]{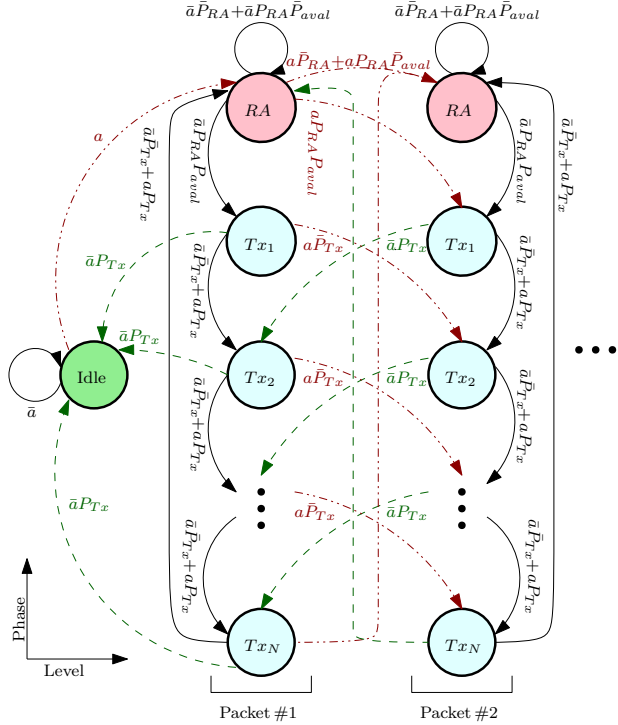}
	\end{center}
	\vspace{-7mm}
	\caption{  Two-dimensional discrete time Markov chain (DTMC) for N-time-slot in SC-UL transmission. $P_{\text{\rm RA}}$ represents the probability of successful {\rm RA}, $P_{\text{\rm Tx}}$ represents the probability of successful SC-UL transmission,  and $P_{\text{\rm aval}}$ represents the probability of available UL frequency. Moreover, the green color indicates empty queue and hence idle state (empty buffer), the red color indicates non-empty queue with {\rm RA} state, and cyan color indicates non-empty queue with scheduled UL transmission state.}\label{N_resources}
\end{figure}
\vspace{-5mm}
\subsection{MAC Layer Parameters}
\vspace{-2mm}
We consider a discrete time slotted network in which time is discretized into slots with equal durations ($T_{\text{\rm s}}$). Each time slot can be used for a single transmission attempt (e.g., RA-SR or EA-Tx for the SC-UL).  Moreover, we assume that a single packet arrival and/or departure can take place per time slot via a First Come First Served (FCFS) discipline. A geometric inter-arrival packet generation, with parameter $a \in [0,1]$ (packet/slot), is assumed at each device. The arrived packets at each device are stored in a buffer (i.e., queue) with infinite length until successful transmission using a UL resource block.\footnote{ Infinite buffers are assumed for generality. In the numerical results, it is shown that small buffer sizes are sufficient as long as the network is stable.} Different from the arrival process, the departure process (i.e., successful packet transmission) cannot be assumed. Instead, the departure process has to be characterized according to the UL transmission protocol and SINR distribution.

\subsection{SC-UL Scheme}
The data is generated at the devices' buffers and the BS is unaware of the devices' buffer status.\footnote{IoT devices have sporadic traffic patterns and can remain idle (i.e., with empty buffers) and go to sleep mode for long periods of time to save battery. Buffer state updates lead to  unnecessary wake-ups that deplete batteries~\cite{Nour}. Furthermore, buffer state updates from billions of devices would impose overwhelming signaling overhead.} Devices with non-empty buffers should send an RA-SR  to the serving BS~\cite{sesia2009lte, Primer_NBIOT}. A device that experiences successful RA-SR is scheduled by the BS for the EA-Tx phase. To model such two-phase UL traffic in 3GPP networks, a tandem (i.e., consecutive) queueing approach with heterogenous departure probabilities is introduced. The tandem queueing model for the SC-UL transmission scheme with joint RA-SR and EA-Tx phases is depicted in Fig.\ref{N_resources}. The first queue (colored in red in Fig.\ref{N_resources}) represents the RA-SR process that occurs prior to resource allocations. The latter queue (colored in cyan in Fig.\ref{N_resources}) represents EA-Tx data transmission after resource allocation.  A detailed description of the two phases for the SC-UL transmission is given in the sequel.


\subsubsection{ \bf RA-SR Phase}
  To request a UL resource block, each device randomly and independently transmits its request on one of the available prime-length orthogonal Zadoff-Chu (ZC) sequences defined by the LTE physical random access channel (PRACH) preamble \cite{sesia2009lte}. Since the number of the ZC sequences is finite, it is possible for more than one device to select the same ZC sequence for RA-SR, which leads to mutual intracell and intercell interference.  Without loss of generality, we assume that all BSs have the same number ($n_{\text{\rm Z}}$) of ZC sequences generated from a single root sequence.\footnote{This implies that the BSs are dense enough such that all the sequences within each BS are generated from cyclic shifts of a single root sequence~\cite{sesia2009lte}.}$^,$\footnote{{ If the neighboring BSs use different root sequences, a thinned PPP with a thinning factor of $1/n_{rs}$ can be used, where $n_{rs}$ is the total number of available root sequences. }}
   We assume a power capture model for successful RA-SR, in which the signal can be decoded if and only if the SINR at the serving BS is greater than the RA-SR  threshold $\theta_{\text{\rm SR}}$. Moreover, each BS can only decode one RA-SR request per ZC code per time slot. That is, across  the intracell interfering devices over the same ZC code within the same BS, only the device that has the highest SINR can succeed, \textcolor{red}{i.e., the BSs have a capturing capability \cite{329345}}. We also assume that the response of each RA-SR attempt is instantaneous and error-free. Upon  RA-SR failure, the ZC code random selection is repeated. It is worth to highlight that the depicted model is consistent with  {\em PRACH configuration index-14} in which there is a PRACH resource in each time slot \cite{sesia2009lte}. {It is assumed that the RA-SR phases are time-synchronized for all the BSs in the network.}

\subsubsection{\bf EA-Tx Phase} Each BS has $q$ resource blocks that are dedicated for UL transmission. A device that {succeeds in} {\rm RA-SR} is granted EA-Tx (i.e., intracell interference free) for the next $N$ subsequent time slots on a randomly selected free resource block by its serving BS.\footnote{The serving BS randomizes the channel allocations for the scheduled devices at each time slot to decorrelate interference across time slots.} If all the $q$ resource blocks are occupied by other UL transmissions, the device has to re-perform the {\rm RA-SR} phase. {If a device flushes its buffer before completing the $N$   {\rm EA-Tx} transmission slots, it immediately goes back to the idle state and releases the allocated channel. Hence, setting $N= 1$ requires a successful RA-SR prior to each {\rm EA-Tx} packet transmission attempt. As such, $N$ is a design parameter for the {\rm SC-UL} scheme.} We assume a power capture model for successful {\rm EA-Tx} transmission, where a data packet departs from the device buffer if and only if the SINR at the serving BS is greater than the UL threshold $\theta_{\text{\rm Tx}}$. We also assume that the feedback of each  transmission attempt is instantaneous and error-free. {It is assumed that the EA-TX phases are time-synchronized in for all the BSs in the network.} 



\subsection{RA-UL Transmission}

 In the RA-UL scheme, each device directly sends the data packets to its closest BS without scheduling. To diversify mutual interference, each device randomly and independently selects a resource block among the $n_{\text{\rm c}}$ available resource blocks for each transmission/retransmission. Since the number of the resource blocks is finite, it is possible for more than one device to utilize the same frequency for RA-UL, which may lead to both intracell and intercell interference. Without loss of generality, we assume that all BSs have the same number of the frequency channels {and the transmission slots are time-synchronized for all the BSs in the network.} 
 We assume a power capture model for successful RA-UL transmission, where a data packet departs from the device buffer if and only if the SINR at the serving BS is greater than the UL threshold $\theta_{\text{\rm UL}}$.  Moreover, in the case of intracell interference, the BS can successfully decode the packet from the device with the highest SINR only, \textcolor{red}{i.e., the BSs have a capturing capability \cite{329345}}. 
 We also assume that the feedback of each  transmission attempt is instantaneous and error-free.  The queueing model for a typical device is shown in  Fig.~\ref{fig_Markov_baseline}, in which the device keeps transmitting as long as it has a non-empty buffer.

 \begin{figure}[t!]
	\centering
	\includegraphics[width=3.75in]{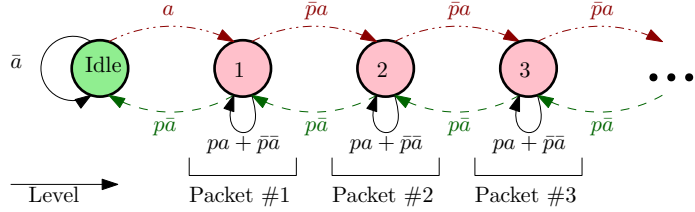}
	\vspace{-7mm}
	\caption{ Queue aware schematic diagram for the direct transmission in RA-UL. $p$ represents the probability of successful transmission, the green color indicates empty queue and hence idle state (not transmitting), and the red color indicates non-empty queue with transmission state. }
	\label{fig_Markov_baseline}
\end{figure}

\subsection{Methodology  of Analysis}

The DTMC in Fig.\ref{N_resources} and Fig.\ref{fig_Markov_baseline} model the temporal evolution of the number of packets in the system as well as the service phase (i.e., RA-SR and {\rm EA-Tx} transmission slots) at each device. Such queueing systems are categorized within the Quasi-Birth-Death (QBD) processes because the population of the system (i.e., the packets in the buffer) can only be incremented or decremented by one in each time slot~\cite{alfa_DTMC}. If the arrival and departure processes are known, the steady-state solution of such QBDs can be characterized. However, the departure process depends on the SINR distribution, which should be derived using stochastic geometry analysis. Meanwhile, the SINR distributions depend on the interfering devices states (e.g., idle or active), which require the steady-state solution of the queueing models. Hence, the stochastic geometry analysis for the SINR distribution and the queueing theory analysis for the devices' steady-state probabilities are interdependent, which necessitates an iterative solution. In this paper, we adopt the following methodology to characterize the performance of the depicted system model with such spatiotemporal interdependence. First, the stochastic geometry analysis is conducted in Section~\ref{sic:Performance_Analysis1}, where expressions for the RA-SR, {\rm EA-Tx}, and the RA-UL success probabilities are obtained as functions queueing theory parameters. Then, the queueing theory analysis is presented in Section~\ref{sic:Performance_Analysis2}, where the steady-state solutions of the queueing models are obtained as functions of the stochastic geometry parameters. The iterative algorithm that simultaneously solves the stochastic geometry and queueing theory set of equations is then presented in Section~\ref{sic:Performance_Analysis3}.

\section{ Performance Analysis} \label{sic:Performance_Analysis}
\vspace{-2mm}
While some IoT devices might be free to move, it is not expected to have tangible variation in terms of the network geometry over consecutive time slots. This is because the locations of the BSs are fixed and that the time slots are too short (e.g., in the scale of milliseconds.) for a tangible geographical displacement. Hence, it is reasonable to consider an arbitrary but fixed network realization that does not change over time. In such static network setup, different devices generally have location dependent performance according to the number and relative locations of proximate interferers~\cite{Meta_Haenggi, Chisci}. However, the  location dependent performance discrepancies are negligible in the depicted system model due to the employed full path-loss inversion power control along with the randomized channel selection per transmission~\cite{Meta_Haenggi_Control, Meta_Elsawy, Gharbieh_tcom}. Consequently, the success probabilities of the typical device (i.e., after spatial averaging) is representative to the performance of all devices, which leads to the following {approximations}:

\vspace{-4mm}
{\begin{approximation} \label{app1}
The departure rates (i.e., transmission success probabilities) of all queues (i.e., devices)\footnote{Hereafter, the terms devices and queues will be used interchangeably. } in the network are assumed to be memoryless and are {characterized} by the departure rate of the typical queue.
\end{approximation}}

\vspace{-4mm}
\begin{remark} \label{rem1}
The full path loss inversion makes the received UL signal power at the serving BSs independent from the service distance (i.e., the distance between the device and the serving BS). Hence, the different realizations of the service distance across the devices do not affect the SINR. Furthermore, the random channel selection randomizes the set of interfering devices over different time slots, which decorrelate the interference across time. Hence, all devices in the network tend to have memoryless departure rates and perform as a typical device as shown in~\cite{Meta_Haenggi_Control, Meta_Elsawy, Gharbieh_tcom}. The accuracy of such approximation is validated in Section~\ref{sec:Results}.
\end{remark}

 Utilizing {Approximation}~\ref{app1}, the performance is assessed for a test BS located at an arbitrary origin in $\mathbb{R}^2$, which becomes a typical BS under spatial averaging. Before delving into the analysis, we state the following two important {approximations} that will be utilized in this paper.

\vspace{-4mm}
 {\begin{approximation} \label{app2}
The spatial correlations between proximate devices, in terms of transmission power and buffer states, can be ignored.
\end{approximation}}

\vspace{-4mm}
\begin{remark} \label{rem2}
It is well known that the sizes of adjacent Voronoi cells are correlated. Such correlation affects the number of devices, as well as, the service distance realizations in adjacent Voronoi cells. Consequently, the transmission powers and devices states (e.g., active or idle) at adjacent cells are correlated. Accounting for such spatial correlation would impede the model tractability. Hence, we follow the common approach in the literature and ignore such spatial correlations when characterizing the aggregate interference~\cite{Meta_Elsawy, Meta_Haenggi_Control, elsawy2014stochastic, marco_uplink, uplink2_jeff, Gharbieh_tcom, 6516885}. However, all spatial correlations are intrinsically accounted for in the Monte Carlo simulations that are used to validate our model in Section~\ref{sec:Results}.
\end{remark}

\vspace{-4mm}
{\begin{approximation} \label{app3}
For {\rm EA-Tx}, the point processes of scheduled inter-cell interfering devices seen at the test BS is {modeled by} a non-homogenous PPP.
\end{approximation}}

\vspace{-4mm}
 \begin{remark} \label{rem3}
Despite that a PPP is used to model the complete set of devices, the subset of UL scheduled devices over a given channel is not a PPP. The constraint of scheduling one device per Voronoi cell per channel leads to a Voronoi-perturbed point process for the set of mutually interfering devices. Approximation~\ref{app3} is commonly used in the literature to maintain tractability~\cite{Meta_Elsawy, elsawy2014stochastic, marco_uplink, uplink2_jeff, Gharbieh_tcom, 6516885}.
\end{remark}

It is worth mentioning that Approximations 1-3 are mandatory for tractability, regularly used in the literature, and are all validated in Section~\ref{sec:Results} via independent Monte-Carlo simulations. Based on these approximations, the spatial and temporal analysis are presented in, respectively, Section~\ref{sic:Performance_Analysis1} and Section~\ref{sic:Performance_Analysis2}. The spatiotemporal model that combines both stochastic geometry and queueing theory analysis is then presented in Section~\ref{sic:Performance_Analysis3}.


\vspace{-5mm}
\subsection{Stochastic Geometry Analysis} \label{sic:Performance_Analysis1}

This section presents the stochastic geometry analysis for the departure rate probabilities for the SC-UL and RA-UL schemes. As mentioned before, the analysis is focused on a test BS located at the origin and a randomly selected device that it serves. For notational convince, let $\bold{h}_\circ$ be the set of channel gains between the test BS and all devices that it serves and let ${h}_\circ \in \bold{h}_\circ$ be the channel gain between the test BS and the selected device for the analysis. For organized presentation, the analysis for each of the SC-UL and RA-UL schemes is provided in a separate subsection.

\subsubsection{SC-UL Scheme}

 Referring to Fig.~\ref{N_resources}, let $x_\circ$ be the probability of being in the idle state (i.e., empty buffer) and $\boldsymbol{\varphi}=[\varphi_{\text{\rm RA}},  \varphi_{\text{\rm Tx}_1}, \varphi_{\text{\rm Tx}_2}, \cdots, \varphi_{\text{\rm Tx}_N} ]$ be the sub-stochastic vector containing the probabilities of being in the RA-SR phase and each of the $N$ phases dedicated for {\rm EA-Tx} for non-empty buffer. For a given $x_\circ$ and $\boldsymbol{\varphi}$, this section characterizes the probabilities $P_{\text{\rm RA}}$, $P_{\text{\rm Tx}}$, and $P_{\text{\rm aval}}$ using stochastic geometry.

First we characterize the {\rm RA-SR} success probability $P_{\text{\rm RA}}$. To evaluate the interference experienced by the test device, we find the Laplace transforms (LT) of the aggregate intracell and intercell interferences. The probability of successful {\rm RA-SR} for the test device is given by
\begin{align} \label{SINR_RA}
P_{\text{\rm RA}}&=\mathbb{P}\left\{  \frac{\rho h_\circ}{\sigma^2 + \mathcal{I}^{(\text{\rm RA})}_{\text{\rm Inter}}+\mathcal{I}^{(\text{\rm RA})}_{\text{\rm Intra}}}>\theta_{\text{\rm SR}},   h_\circ> h_i; \; \forall h_i \in \bold{h}_o \setminus h_\circ   \right\}, \notag \\
&=\mathbb{P}\left\{ \frac{\rho h_\circ}{\sigma^2 + \mathcal{I}^{(\text{\rm RA})}_{\text{\rm Inter}}+\mathcal{I}^{(\text{\rm RA})}_{\text{\rm Intra}}}>\theta_{\text{\rm SR}} \;\mid  h_\circ> h_i \;  \forall h_i \in \bold{h}_o \setminus h_\circ\right\}\mathbb{P}\left\{h_\circ> h_i \;  \forall h_i \in \bold{h}_o \setminus h_\circ \right\}
\end{align}

\noindent which follows from the adopted SINR capture model along with the fact that the BS can only decode the RA-SR with the highest SINR among the intracell interfering devices. ${\mathcal{I}^{(\text{\rm RA})}_{\text{\rm Inter}}}$ denotes intercell interference for RA-SR while  ${\mathcal{I}^{(\text{\rm RA})}_{\text{\rm Intra}}}$ denotes the intracell interference  for RA-SR, and $\sigma^2$ is the noise power. The {\rm RA-SR}  success probability in \eqref{SINR_RA} is characterized the following lemma.

\begin{lemma}\label{lemma_RA_success}
The {\rm RA-SR} success probability  in a PPP single-tier network where each device employs full path-loss inversion power control is given by {
\begin{align}\label{eq:RA_success}
P_{\text{\rm RA}}&= \mathbb{E}_{\mathcal{N}} \left\{\frac{\sum\limits_{k=1 }^{\mathcal{N}+1} \binom{\mathcal{N}+1}{k}(-1)^{k+1} \exp \left\{\frac{-k \; \theta_{\text{\rm RA}\; }\sigma^2}{\rho}\right\}\mathscr{L}_{\mathcal{I}^{(\text{\rm RA})}_{\text{\rm Inter}}}\left(\frac{k \; \theta_{ \text{\rm RA}}}{\rho} \right) \mathscr{L}_{\mathcal{I}^{(\text{\rm RA})}_{\text{\rm Intra}}\mid_{\mathcal{N}\!=\mathtt{n}}} \left(\frac{k \; \theta_{ \text{\rm RA}}}{\rho} \right)}{\mathcal{N}+1}\right\},
\end{align}
with
\begin{align}
\mathscr{L}_{\mathcal{I}^{(\text{\rm RA})}_{\text{\rm Inter}}}\left(\frac{k \; \theta_{ \text{\rm RA}}}{\rho}\right) &\approx \exp\left(-2 \; k \; \theta_{ \text{\rm RA}}\; \frac{ \varphi_{\text{\rm RA}}  \;\tilde{\mu}}{\lambda} \; \frac{{}_2F_1\left(1,1-\frac{2}{\eta},2-\frac{2}{\eta},-k\;\theta_{\text{\rm SR}}\right)}{\eta-2} \right), \label{Eq:laplase_Inter}\\
\mathscr{L}_{\mathcal{I}^{(\text{\rm RA})}_{\text{\rm Intra}}\mid_{\mathcal{N}\!=\mathtt{n}}}\left(\frac{k \; \theta_{ \text{\rm RA}}}{\rho}\right)&=\frac{\mathtt{n}+1}{1+k \; \theta_{\text{\rm SR}}} \left(\frac{1}{\mathtt{n}} -\frac{\Gamma(\mathtt{n}) \;\Gamma(2+k \; \theta_{\text{\rm SR}})}{\Gamma(2+\mathtt{n}+k \; \theta_{\text{\rm SR}})} \right),
\end{align}}
where $\tilde{\mu}=\frac{\mu}{n_{\text{\rm Z}}}$, ${}_2 F_1(.)$ is the Gaussian hypergeometric function, and \eqref{Eq:laplase_Inter} is not exact due to Approximation~2 { mentioned earlier in this section}. The expectation $\mathbb{E}_{\mathcal{N}}\{\cdot\}$ is over the Probability Density Function (PDF) of the number of neighbors  $\mathcal{N}$ which is found in \cite{6576413} as:
\begin{align}\label{pdf_users}
\mathbb{P}\{\mathcal{N} = \mathtt{n}\} \approx \frac{\Gamma(\mathtt{n}+c)}{\Gamma(\mathtt{n}+1)\Gamma(c)} \frac{(\varphi_{\text{\rm RA}}  \;\tilde{\mu})^{\mathtt{n}} (\lambda c)^c}{(\varphi_{\text{\rm RA}}  \;\tilde{\mu}+\lambda c)^{\mathtt{n}+c}},
\end{align}
where  $\varphi_{\text{\rm RA}} $ is the probability of being in the {\rm RA-SR} state and $c=3.575$ is a constant related to the approximate PDF of the PPP Voronoi cell area in $\mathbb{R}^2$.  

\begin{proof}
See Appendix \ref{sec:AppA}.
\end{proof}
\end{lemma}

{After a successful {RA-SR}, the device proceeds to the {\rm EA-Tx} phase for data transmission if there is an available (i.e., not used by other devices) resource block at the serving BS. Let $\varphi_{\text{\rm Tx}_i}$ be the probability that a device is using one of the available $q$ resource blocks for data transmission and is in the $1<i^{th}<N$ transmission time slot.  Let $\mathcal{N}_i$ be a random variable that represents the number of transmitting devices  at the $i^{th}$ transmission time slot. As a result, the PDF of the number of devices within the same cell that have reserved resource blocks for the next time slot is given by
\begin{align}\label{pdf_users_Tx}
\mathbb{P}\left\{\sum\limits_{i=1 }^{N-1} \mathcal{N}_i= \mathtt{n}\right\} \approx \frac{\Gamma(\mathtt{n}+c)}{\Gamma(\mathtt{n}+1)\Gamma(c)} \frac{(\mu \sum_{k=1}^{N-1}\varphi_{\text{\rm Tx}_k} )^{\mathtt{n}} (\lambda c)^c}{(\mu \sum_{k=1}^{N-1}\varphi_{\text{\rm Tx}_k} +\lambda c)^{\mathtt{n}+c}},
\end{align}
where \eqref{pdf_users_Tx} follows from the superposition property of the PPP. And hence, the intensity of the device that have reserved resource blocks for the next time slot equals to $(\mu \sum_{k=1}^{N-1}\varphi_{\text{\rm Tx}_k}) $. Consequently, the probability that a device proceeds from the RA-SR phase to {\rm EA-Tx} phase, i.e., the probability of finding an available resource block at the next time slot for {\rm EA-Tx} transmission can be expressed as follows:
\begin{align}\label{eq:Paval}
P_{\text{\rm aval}}&= \mathbb{P}\left\{\sum\limits_{i=1 }^{N-1} \mathcal{N}_i < q\right\} =  \sum_{n=1}^{q-1} \frac{\Gamma({n}+c)}{\Gamma({n}+1)\Gamma(c)} \frac{(\mu \sum_{k=1}^{N-1}\varphi_{\text{\rm Tx}_k} )^{{n}} (\lambda c)^c}{(\mu \sum_{k=1}^{N-1}\varphi_{\text{\rm Tx}_k} +\lambda c)^{{n}+c}},
\end{align}
where \eqref{eq:Paval} follows from the definition of the CDF for a random variable. It is worth to mention that \eqref{pdf_users_Tx} and \eqref{eq:Paval} exclude the devices in the $N^{th}$ transmission slot since those devices have to resend RA-SR and will not be allocated {\rm EA-Tx} resources in the next time slot.}

Once the device enters the {\rm EA-Tx} phase, it operates over an exclusive channel for $N$ subsequent time slots. Exploiting Approximation 1 and the fact that the channel allocation randomly changes from one time slot to another, the packet transmission success probability is independent from one time slot to another. Hence, the probability of {\rm EA-Tx} packet transmission success in an arbitrarily selected time slot is
\begin{align} \label{SINR_EU}
P_{\text{\rm Tx}}&=\mathbb{P}\left\{  \frac{\rho h_\circ}{\sigma^2 + \mathcal{I}^{(\text{\rm EA})}_{\text{\rm Inter}}}>\theta_{\text{\rm Tx}}   \right\}
\end{align}
Comparing \eqref{SINR_EU} and \eqref{SINR_RA}, it is clear the {\rm EA-Tx} does not experience intracell interference, and hence, the condition of highest SINR within the cell does not exits. Furthermore, it is important to note that   $\mathcal{I}^{(\text{\rm RA})}_{\text{\rm Inter}}$  statistically dominates $\mathcal{I}^{(\text{\rm EA})}_{\text{\rm Inter}}$. This is because {\rm EA-Tx} allows at most one intercell interferer per BS per channel as opposed to RA-SR which permits multiple intercell interferers per BS per channel. The packet transmission success probability in the {\rm EA-Tx} phase is characterized in the following lemma.

\begin{lemma}\label{lemma_Tx_success}
The probability of successful data transmission in a PPP single-tier network where each device employs full path-loss inversion power control, can be expressed as
\vspace{-2mm}
\begin{align}\label{eq:Out2}
  {P_{\rm Tx}} &\; {\approx} \; \exp\left\{- \frac{\sigma^2 \theta_{\text{\rm Tx}}}{\rho} -  \frac{2  \theta_{\text{\rm Tx}}}{(\eta-2) } \;{}_2F_1\left(1,1-\frac{2}{\eta},2-\frac{2}{\eta},-\theta_{\text{\rm Tx}}\right) \right\} \notag \\
& \overset{(\eta=4)}{\approx} \exp\left\{- \frac{\sigma^2 \theta_{\text{\rm Tx}}}{\rho} - \sqrt{\theta_{\text{\rm Tx}}}\arctan\left({\sqrt{\theta_{\text{\rm Tx}}}}\right)   \right\}.
\end{align}
\begin{proof}
Similar to \cite[Theorem 1]{uplink_alamouri}, {where \eqref{eq:Out2} is not exact due to Approximations~2~and~3 mentioned earlier in this section}.
\end{proof}
\end{lemma}

\subsubsection{RA-UL Scheme} Referring to Fig.~\ref{fig_Markov_baseline}, let $x_\circ$ be the probability of empty buffer. This section characterizes the successful packet transmission probability  $p$ for a given $x_\circ$ using stochastic geometry.

In the RA-UL scheme, the devices directly transmit UL data packets over a randomly selected channel. Hence, transmissions in the RA-UL scheme are subject to intracell and intercell interference. Considering the fact that only one of the intracell interfering devices can succeed  at a given time slot, the transmission success probability $p$ for the RA-UL scheme is characterized via the following lemma.

\begin{lemma} \label{lem:out_NB_IoT}
	The transmission success probability in the depicted PPP network with RA-UL scheme, where the message with the highest SINR is decodable if its SINR is greater than the detection threshold $\theta_{\text{\rm UL}}$, is given by

{
\begin{align}\label{eq:TX_NB_success}
p&= \mathbb{E}_{\mathcal{N}} \left\{\frac{\sum\limits_{k=1 }^{\mathcal{N}+1} \binom{\mathcal{N}+1}{k}(-1)^{k+1} \exp \left\{\frac{-k \; \theta_{\text{\rm UL}} \sigma^2}{\rho}\right\}\mathscr{L}_{\mathcal{I}_{\rm Inter}}\left(\frac{k \; \theta_{\text{\rm UL}}}{\rho} \right) \mathscr{L}_{\mathcal{I}_{\rm Intra}\mid_{\mathcal{N}=\mathtt{n}}} \left(\frac{k \; \theta_{\text{\rm UL}}}{\rho} \right)}{\mathcal{N}+1}\right\},
	\end{align}
	with
	\begin{align}
	\mathscr{L}_{\mathcal{I}_{\text{\rm Inter}}}\left(\frac{k \; \theta_{\text{\rm UL}}}{\rho}\right) &\approx \exp\left(-2 \; k \; \theta_{\text{\rm UL}}\; \frac{ \bar{x}_\circ \;\mu^{\prime}}{\lambda} \; \frac{{}_2F_1\left(1,1-\frac{2}{\eta},2-\frac{2}{\eta},-k\;\theta_{\text{\rm UL}}\right)}{\eta-2} \right), \label{Eq:laplase_Inter11}\\
	\mathscr{L}_{\mathcal{I}_{\rm Intra}\mid_{\mathcal{N}=\mathtt{n}}}\left(\frac{k \; \theta_{\text{\rm UL}}}{\rho}\right)&=\frac{\mathtt{n}+1}{1+k \; \theta_{\text{\rm UL}}} \left(\frac{1}{\mathtt{n}} -\frac{\Gamma(\mathtt{n}) \;\Gamma(2+k \; \theta_{\text{\rm UL}})}{\Gamma(2+\mathtt{n}+k \; \theta_{\text{\rm UL}})} \right),
	\end{align}}
	where { $\mu^{\prime}=\mu/n_{\text{\rm c}}$ and \eqref{Eq:laplase_Inter11} is not exact due to Approximation~2 mention earlier in this section. The expectation $\mathbb{E}_{\mathcal{N}}\{\cdot\}$ is over the PDF of the number of neighbors  $\mathcal{N}$ as: 
		\begin{align}\label{pdf_users1}
		\mathbb{P}\{\mathcal{N} = \mathtt{n}\} \approx \frac{\Gamma(\mathtt{n}+c)}{\Gamma(\mathtt{n}+1)\Gamma(c)} \frac{(\bar{x}_\circ  \;\mu^{\prime})^{\mathtt{n}} (\lambda c)^c}{(\bar{x}_\circ   \;\mu^{\prime}+\lambda c)^{\mathtt{n}+c}},
		\end{align}
 where $x_\circ$ is the probability of being in the idle state and $c=3.575$ is a constant related to the approximate PDF of the PPP Voronoi cell area in $\mathbb{R}^2$ \cite{6576413}.   }

\begin{proof}
	Similar to Lemma~\ref{lemma_RA_success}.
\end{proof}

\end{lemma}

\subsection{Queueing Theory Analysis} \label{sic:Performance_Analysis2}



This section develops the queueing theory analysis to track the temporal packet accumulation/departure at the devices' buffers. As stated in Approximation~\ref{app1} and Remark~\ref{app1}, the steady-state distribution of the test device is representative to all devices in the system and that the queue departures are memoryless. Since only one packet arrival and/or departure can occur in each time slot, the temporal evolution of the number of packets in the test device buffer can be traced  via a QBD queueing model with the following general probability transition matrix
\begin{align}
\mathbf{P}=\begin{bmatrix}
B & \mathbf{C} &  &  &  & \\
\mathbf{E} & \mathbf{A}_1 & \mathbf{A}_0 &  &  &\\
  &\mathbf{A}_2  & \mathbf{A}_1 & \mathbf{A}_0 &   &\\
  &  & \mathbf{A}_2 & \mathbf{A}_1 & \mathbf{A}_0&  &\\
  &  &     &\ddots  &\ddots  &\ddots
\end{bmatrix},
\label{eq:trans_matrix1}
\end{align}
\noindent where $B \in\mathbb{R} ,\mathbf{C} \in \mathbb{R}^{1\times n}, \mathbf{E} \in \mathbb{R}^{n \times1}, \mathbf{A}_0 \in \mathbb{R}^{n \times n}, \mathbf{A}_1 \in \mathbb{R}^{n \times n}$, and $\mathbf{A}_2 \in \mathbb{R}^{n \times n}$ are the sub-stochastic matrices that capture the transitions between the queue levels (i.e., number of packets in the buffer). Particularly, the sub-stochastic matrices  $\mathbf{A}_0$, $\mathbf{A}_1$ and $\mathbf{A}_2$  capture, respectively, one step increasing, unchanged, and one step decreasing number of packets int the buffer. Furthermore, $B$, $\mathbf{C}$, and $\mathbf{E}$ are the boundary transition vectors between the idle state and Level 1 (i.e., empty buffer and non-empty buffer with only one packet). The transmission matrix in \eqref{eq:trans_matrix1} will be populated according to the utilized transmission scheme (i.e., SC-UL or RA-UL) as shown in the sequel.

\subsubsection{SC-UL Scheme}

To analyze the QBD queueing model for the SC-UL, shown in Fig.\ref{N_resources}, we employ the matrix-analytic-method (MAM)~\cite{MAM, alfa_DTMC}. Particularly, the departure process is modeled via a Phase (PH) type distribution that captures all transition phases the queue can experience until a packet departure.\footnote{Interested readers may refer to \cite{Gharbieh_tcom} for full technical details on how to combine MAM and stochastic geometry into a unified framework.} Referring to the SC-UL queueing model shown in Fig.\ref{N_resources} and utilizing the PH distribution for the departure process,  the transmission matrix in \eqref{eq:trans_matrix1} is populated as follows: $B=\bar{a}$, $\mathbf{C} \! = \!\{a,\boldsymbol{0}\}$, $\mathbf{E}  = \bar{a} \mathbf{s}$, $\mathbf{A}_0 ={a} \mathbf{S}$, $\mathbf{A}_1 =a \mathbf{G} +\bar{a} \mathbf{S}$, and $\mathbf{A}_2 =\bar{a} \mathbf{G}$ where $\mathbf{0}$ is a row vector of zeros of length  {$N$}, $\mathbf{s} =\mathbf{e}-\mathbf{S} \times \mathbf{e}$, and $\mathbf{e}$ is used in this paper to represent a column vector of the proper length. The sub-stochastic matrices $S$ and $G$ are of size $(N+1) \times (N+1)$ and are given by
\begin{align}
\mathbf{S}\!=\!\begin{bmatrix}
\bar{P}_{\text{\rm RA}}+P_{\text{\rm RA}}\bar{P}_{\text{aval}}&P_{\text{\rm RA}}P_{\text{\rm aval}}& 0  & 0 &\hdots  &0\\
 0 & 0 &\bar{P}_{\text{\rm Tx}}& 0 &\hdots  &0\\
 0& 0 & 0 &\bar{P}_{\text{\rm Tx}} &\hdots  &0\\
 \vdots & \vdots&  \vdots& \ddots & \ddots  &\vdots \\
 0 &  0&  0& 0 &\hdots  &\bar{P}_{\text{\rm Tx}}\\
\bar{P}_{\text{\rm Tx}}&  0&  0& 0 &\hdots &0
\end{bmatrix}, \quad\mathbf{G}\!=\!\begin{bmatrix}
0 &0& 0  & 0 &\hdots  &0\\
0 & 0 &P_{\text{\rm Tx}}& 0 &\hdots  &0\\
0& 0 & 0 &P_{\text{\rm Tx}} &\hdots  &0\\
\vdots & \vdots&  \vdots& \ddots & \ddots  &\vdots \\
0 &  0&  0& 0 &\hdots  &P_{\text{\rm Tx}}\\
P_{\text{\rm Tx}}&  0&  0& 0 &\hdots &0
\end{bmatrix}.
\label{S_G_matrix}
\end{align}
where  ${P}_{\text{\rm RA}}$, $P_{\text{\rm RA}}$, and $P_{\text{\rm Tx}}$ are derived, respectively in \eqref{eq:RA_success}, \eqref{eq:Paval}, and \eqref{eq:Out2} via stochastic geometry analysis. {It is worth mentioning that the probability transition matrix of the QBD queueing model is irreducible, aperiodic, and positive recurrent (i.e., ergodic DTMC. Hence, the steady-state probabilities exist \cite{alfa_DTMC}.}

By Leoynes theorem \cite{loyens}, the queueing model in \eqref{eq:trans_matrix1}  is stable if the average arrival rate is less than the average service rate. Following \cite{alfa_DTMC,MAM}, let $\mathbf{A}=\mathbf{A}_0+\mathbf{A}_1+\mathbf{A}_2$ and the vector $\boldsymbol{\nu}$ be the unique solution of the system of equations given by $\boldsymbol{\nu}  \mathbf{A}=\boldsymbol{\nu}  $ and $ \boldsymbol{\nu}  \mathbf{e}=1$. Then, the sufficient stability condition for the SC-UL is
\begin{align}\label{stability_condition}
\boldsymbol{\nu} \bar{a}\mathbf{G}\mathbf{e}>\boldsymbol{\nu}  a\mathbf{S}  \mathbf{e}.
\end{align}

Let $\mathbf{x}=[x_\circ, \mathbf{x}_1, \mathbf{x}_2,\ \hdots,\ ]$ be the stationary distribution, where $x_\circ$ represents the probability of being in the idle state and $\mathbf{x}_i = [x_{i,1} x_{i,2} \dots x_{i,N}]$ is the probability vector of being in any service phases at the level $i$ of the queue. The steady-state solution for the SC-UL queueing model is obtained by solving the following system of equations
\begin{align} \label{hhss}
\mathbf{x}  \mathbf{P}=\mathbf{x} \quad  \text{and} \quad\mathbf{x}  \mathbf{e}&=1.
\end{align}
An explicit expression for the steady-state solution vector $\mathbf{x}$ can be obtained as highlighted in the following lemma.

\begin{lemma} \label{lemma_dis}
 Given that the stability condition in \eqref{stability_condition} is satisfied, then solving the system of equations in \eqref{hhss} for the SC-UL scheme gives the following steady-state solution
\begin{align}
x_\circ & = \left( 1+\mathbf{C}  \left(  [\mathbf{I}-  a \mathbf{G} - \bar{a} \mathbf{S} - \mathbf{R} \bar{a} \mathbf{G} ] [\mathbf{I}-\mathbf{R}] \right)^{-1}  \mathbf{e} \right)^{-1} \notag \\ \text{and}\quad \mathbf{x}_i & =\left\{\begin{matrix}
 x_\circ \mathbf{C} [\mathbf{I}-  a \mathbf{G} - \bar{a} \mathbf{S} - \mathbf{R} \bar{a} \mathbf{G} ]   &  i=1\\
  {\mathbf{x}_1} \mathbf{R}^{i-1}  &  i>1
\end{matrix}\right.,
\label{eq_ramp}
\end{align}
\noindent where $\mathbf{R}$ is the \rm{MAM} $\mathbf{R}$ matrix derived via \textbf{Algorithm~\ref{R_matrix_algo}}.
\end{lemma}
\begin{proof}
 $x_\circ$ and $\mathbf{x}_1$ are obtained by solving the boundary equation $\mathbf{x}_1=x_\circ\mathbf{C}+\mathbf{x}_1(\mathbf{A}_1+\mathbf{R}\mathbf{A}_2)$ and normalization condition $x_\circ+\mathbf{x}_1 [\mathbf{I}-\mathbf{R}]^{-1}\mathbf{e}=1$, where $\mathbf{A}_1$ and $\mathbf{A}_2$ are defined in \eqref{eq:trans_matrix1}. Then $\mathbf{x}_i$ in \eqref{eq_ramp} follows from the definition of the $\mathbf{R}$ matrix~\cite{MAM, alfa_DTMC}, which is the minimal non-negative solution of $\mathbf{R}=\mathbf{A}_0+\mathbf{R}\mathbf{A}_1+\mathbf{R}^2\mathbf{A}_2$.
 \end{proof}
\begin{remark}
Since neither $\mathbf{S}$ nor $\mathbf{G}$ is a rank one matrix, the \rm{MAM} $\mathbf{R}$ matrix cannot be expressed via an explicit expression and is determined via the numerical algorithm given in \textbf{Algorithm~\ref{R_matrix_algo}}.
\end{remark}

\begin{figure}
\begin{algorithm}[H]
\small
Initialize $\mathbf{R}_{[0]}$ with zeros.\\
\While { $\left|\mathbf{R}_{[m]} -\mathbf{R}_{[m-1]} \right| \geq \epsilon $  } {
1- Calculate $\mathbf{R}_{[m+1]}$, using $\mathbf{R}_{[m+1]}=\mathbf{A}_0+\mathbf{R}_{[m]}\mathbf{A}_1+\mathbf{R}_{[m]}^2\mathbf{A}_2$.\\
2- Increment $m$.}
  return $\mathbf{R} \leftarrow \mathbf{R}_{[m]}$.		
 \caption{Numerical computation of the \rm{MAM} $\mathbf{R}$ matrix.}
 \label{R_matrix_algo}
 \normalsize
\end{algorithm}
\end{figure}

Using the steady-state solution in Lemma~\ref{lemma_dis}, the vector $\boldsymbol{\varphi} = [{\varphi}_{\text{\rm RA}}, {\varphi}_{\text{\rm Tx}_1} , {\varphi}_{\text{\rm Tx}_2}, \cdots,  {\varphi}_{\text{\rm Tx}_N}]$, which contains the probability of being in each phase regardless of the queue state, can be obtained as
\begin{align}
\boldsymbol{\varphi}&=\mathbf{x}_1 + \mathbf{x}_2 + \mathbf{x}_3 + \cdots = \mathbf{x}_1 + \mathbf{x}_1 \mathbf{R}+ \mathbf{x}_1 \mathbf{R}^{2} + \cdots \notag \\
&=\mathbf{x}_1[\mathbf{I}-\mathbf{R}]^{-1},
\label{eq_marg}
\end{align}
where \eqref{eq_marg} follows from the fact that $\mathbf{R}$ has a spectral radius less than one  \cite{alfa_DTMC}. 

\subsubsection{RA-UL Scheme} For the RA-UL scheme shown in Fig.~\ref{fig_Markov_baseline}, the queueing model in  \eqref{eq:trans_matrix1} is populated as follows: $B=\bar{a}$, $\mathbf{C} \! = a$, $\mathbf{E}  = \bar{a} p$, $\mathbf{A}_0 ={a} \bar{p}$, $\mathbf{A}_1 =\bar{a} \bar{p} + a p $, and $\mathbf{A}_2 =\bar{a} p$. Hence,  the RA-UL queueing model reduces to a Geo/Geo/1 queue with the following transition matrix
\vspace{-2mm}

\begin{align}
\mathbf{P}=\begin{bmatrix}
\bar{a} & a &  &  &  & \\
\bar{a} p & \bar{a} \bar{p} + ap & a\bar{p} &  &  &\\
&\bar{a} p & \bar{a} \bar{p} + ap & a\bar{p} &   &\\
&  & \bar{a} p & \bar{a} \bar{p} + ap & a\bar{p} &  &\\
&  &     &\ddots  &\ddots  &\ddots
\end{bmatrix}.
\label{NB_MAT}
\end{align}

The sufficient stability condition for the RA-UL scheme with the transition matrix in \eqref{NB_MAT} reduces from \eqref{stability_condition} to
\begin{align} \label{condition_NB}
\bar{a}p >  a\bar{p} .
\end{align}

Given that the stability condition in \eqref{condition_NB} is satisfied, let $\mathbf{x}=[x_\circ,\ x_1,\ x_2,\ \hdots,\ ]$ be the steady-state distribution of RA-UL queueing model, where $x_i$ represents the probability of having $i$ packets in the buffer.  Following the same procedure in Lemma~\ref{lemma_dis}, solving the system of equations $\mathbf{x}  \mathbf{P}=\mathbf{x}$ and $\mathbf{x}  \mathbf{e}=1$ with the transition matrix in \eqref{NB_MAT} yields to the following steady-state solution
\begin{align}
\begin{matrix}
x_i=R^i \frac{x_\circ}{\bar{p}}, & \text{where} \quad R=\frac{a\bar{p}}{ \bar{a}p} & \text{and} \quad x_\circ=\frac{p-a}{p}.
\end{matrix}
\label{eq_NB}
\end{align}

\subsection{Iterative Solution \& Performance Assessment }\label{sic:Performance_Analysis3}
Sections~\ref{sic:Performance_Analysis1} and \ref{sic:Performance_Analysis2} show clear interdependence between stochastic geometry (i.e., spatial) and queueing theory (i.e., temporal) analysis. Particularly, the steady-state solution ${x}_\circ$ and $\boldsymbol{\varphi}$ are required in Section~\ref{sic:Performance_Analysis1} to characterize the interference and derive the packet departure rates via stochastic geometry. Meanwhile, the packet departure rates are required to derive the DTMC steady-state solution in Section~\ref{sic:Performance_Analysis2} via queueing theory. Hence, we employ an iterative solution to simultaneously solve the system of equations obtained via stochastic geometry and queueing theory analysis. By virtue of the fixed point theorem, such iterative solutions are shown to converge to a unique solution~\cite{Heanggi_Ra1, Zhuang, Gharbieh_tcom, Chisci, interacting_queues2, SarElaFouNou:07}. The output of the iterative spatiotemporal algorithm provides the steady-state probabilities and packet departure rates for the considered transmission scheme, which are then used to define the following performance metrics.
\begin{itemize}
\item \textbf{Average buffer Size $\mathbb{E}\left\{\mathcal{Q}_{\text{\rm L}}\right\}$:}  Let $\mathcal{Q}_{\text{\rm L}}$ be the instantaneous buffer size of the test device, then the average buffer size in given by
 \begin{align} \label{ave_queue}
\mathbb{E}\left\{\mathcal{Q}_{\text{\rm L}}\right\} &=\!\! \sum_{\mathtt{n}=2}^{\infty} (\mathtt{n}-1) \mathbb{P} \left\{\mathcal{Q}_{\text{\rm L}}=\mathtt{n} \right\} \!=\!\sum_{\mathtt{n}=2}^{\infty} (\mathtt{n}-1) \sum_{j} x_{\mathtt{n},j},
\end{align}
where $x_{\mathtt{n},j}$ denotes the probability of being in level $\mathtt{n}$ (i.e., the buffer contain $i$ packets) and phase $j$.

\item {\textbf{ Waiting Time in the Queue}: Let $\mathcal{W_{\text{q}}}$ be the queueing delay (i.e., the number of time slots spent in the buffer until uplink scheduling) experienced by a given packet, then the average delay ($\mathbb{E}\left\{\mathcal{W_{\text{q}}}\right\}$), the variance ($\text{Var}\left(\mathcal{W_{\text{q}}}\right)$), and the index of dispersion ($D$) can be evaluated, respectively, as:
\begin{align}
\mathbb{E}\left\{\mathcal{W_{\text{q}}}\right\}&=\sum_{j=1}^{\infty}j \mathbb{P}\left\{\mathcal{W_{\text{q}}}=j \right\},\label{wait_mean}\\
\text{Var}\left(\mathcal{W_{\text{q}}}\right)&=\sum_{j=1}^{\infty}\left(j-\mathbb{E}\left\{\mathcal{W_{\text{q}}}\right\}\right)^2 \; \mathbb{P}\left\{\mathcal{W_{\text{q}}}=j \right\},\label{wait_var}\\ 
D&=\frac{\text{Var}\left(\mathcal{W_{\text{q}}}\right)}{\mathbb{E}\left\{\mathcal{W_{\text{q}}}\right\}}. \label{ind_dis}
\end{align}}
\vspace{-5mm}
\item \textbf{Stability Region}:  Defines the system parameters  where the stability conditions \eqref{stability_condition} and \eqref{condition_NB} are satisfied for, respectively, the SC-UL and RA-UL schemes. Operating within the stability region guarantees bounded average delay. Otherwise,  the network fails to satisfy the spatiotemporal traffic requirements of the IoT devices, in which the average delay and the average number of packets in the buffers become infinite.
\end{itemize}

The spatiotemporal iterative solution and performance of each transmission scheme are presented in the sequel.

\subsubsection{SC-UL Scheme} The spatiotemporal iterative algorithm for the SC-UL scheme is given in the following theorem.

\begin{theorem}
\label{theorem1}
The probability  of being idle and the steady-state sub-stochastic vector $\boldsymbol{\varphi}$ for the transmission phases for a generic device in the SC-UL scheme is obtained via {\bf Algorithm \ref{dist_algo}}.

\begin{figure}
\begin{algorithm}[H]
\small
Initialize   $x_\circ$ and $\boldsymbol{\varphi}$ such that $x_\circ + \boldsymbol{\varphi} \times \mathbf{e} = 1$.\\
\While { $\left|\boldsymbol{\varphi}_{[m]} - \boldsymbol{\varphi}_{[m-1]} \right| \geq \epsilon $  } {
1- Evaluate $P_{\text{\rm RA}}$, $P_{\text{\rm Tx}}$, and $P_{\text{\rm aval}}$ in  \eqref{eq:RA_success}, \eqref{eq:Out2}, and \eqref{eq:Paval} respectively,  using $\boldsymbol{\varphi}_{[m-1]}$.\\
2- Construct $\mathbf{S}$ and $\mathbf{G}$ using $P_{\text{\rm RA}}$ and $P_{\text{\rm Tx}}$ as in \eqref{S_G_matrix}. \\
3- Construct $\boldsymbol{\nu} = [\nu_1, \nu_2, \dots, \nu_{N+1}]$ such that $\nu_1= 1/(1+N \; P_{\text{\rm RA}}P_{\text{\rm aval}}) $ and $\nu_l= \nu_1 P_{\text{\rm RA}}P_{\text{\rm aval}}$ for $l\in\{2,\dots,N+1\}$.\\
4- Check the stability condition in \eqref{stability_condition}. \\
\eIf{$\boldsymbol{\nu} \bar{a}\mathbf{G}\mathbf{e}>\boldsymbol{\nu}  a\mathbf{S}  \mathbf{e}$}{
  Calculate $\mathbf{R}_{[m]}$ from Algorithm 1.\\
   Calculate $\mathbf{x}_1$ from Lemma~\ref{lemma_dis} and $\boldsymbol{\varphi}_{[m]}$ from \eqref{eq_marg}.\;
   }{ return unstable network \\
Terminate the algorithm.}
  5- Increment $m$.}
  Return $\mathbf{R} \leftarrow \mathbf{R}_{[m]}$, $\boldsymbol{\varphi} \leftarrow \boldsymbol{\varphi}_{[m]}$,  and $x_\circ \leftarrow 1-\boldsymbol{\varphi}_{[m]}\mathbf{e} $.		
 \caption{Computation of $x_\circ$ and $\boldsymbol{\varphi}$ for SC-UL Scheme.}
 \label{dist_algo}
 \normalsize
\end{algorithm}
\end{figure}
\begin{proof}
The proof follows from the RA success probability in Lemma~\ref{lemma_RA_success}, the UL transmission success probability in Lemma~\ref{lemma_Tx_success}, the steady-state distribution in Lemma~\ref{lemma_dis}, and the numerical computation of the $\mathbf{R}$  matrix in \bf{Algorithm~\ref{R_matrix_algo}}.
\end{proof}
\end{theorem}

Using the output of {\bf Algorithm \ref{dist_algo}}, the average queue length for a stable device is given by
\begin{align} \label{ave_queue_power}
\mathbb{E}\left\{\mathcal{Q}_{\text{\rm L}}\right\} &= (\mathbf{x}_2 + 2\mathbf{x}_3 + 3\mathbf{x}_4 + \cdots) \; \bold{e} = \mathbf{x}_2(1 +2\bold{R} + 3\bold{R} + \cdots)\; \bold{e}\notag \\
&=\bold{x}_1 \bold{R} (\bold{I}-\bold{R})^{-2}\bold{e},
\end{align}
where \eqref{ave_queue_power} follows from the fact that $\mathbf{R}$ has a spectral radius less than one  \cite{alfa_DTMC}.

{Moreover, the waiting time in the queue is given by \cite{alfa_DTMC}:
\begin{align}\label{wait_pdf}
\mathbb{P} \left\{\mathcal{W_{\text{q}}}=0 \right\} &=x_\circ, \; \text{and} \quad \mathbb{P}\left\{\mathcal{W_{\text{q}}}=j \right\} =\displaystyle{\sum\limits_{v=1}^{j} \mathbf{x}_v \mathbf{B}_j^{(v)}}, \; \text{with} \notag  \\
\quad \mathbf{B}^{(k)}_j  &= \left\{\begin{matrix}
\textbf{S}^{j-1} \mathbf{G}  & k=1, j \geq 1  \\
\mathbf{G}^k  &   j=k, k\geq 1 \\
\mathbf{S} \; \mathbf{B}^{(k)}_{j-1}+\mathbf{G} \;\mathbf{B}^{(k-1)}_{j-1}  & k \geq j\geq 1 \\
\end{matrix}\right..
\end{align}
The average value, the variance, and the index of dispersion can be computed by substituting \eqref{wait_pdf} in \eqref{wait_mean},  \eqref{wait_var}, and  \eqref{ind_dis}, respectively. }


\begin{figure}
	\begin{algorithm}[H]
		\small
		Initialize $\mathbf{x}_{[0]}$.\\
		\While { $\left|\mathbf{x}_{[m]} - \mathbf{x}_{[m-1]} \right| \geq \epsilon $  } {
		    1- Calculate $p$ using $\mathbf{x}_{[m-1]}$ and Eq. \eqref{eq:TX_NB_success}.\\
			2- Update the transition matrix $\mathbf{P}$ in \eqref{NB_MAT} with $p$.  \\
		    3- Check the stability condition in \eqref{condition_NB}. \\
			\eIf{$\bar{a} p > a \bar{p}$}{
			Obtain $\mathbf{x}_{[m]}$ by solving the queueing system in Eq. \eqref{eq_NB}.\;
			}{ return unstable network \\
				Terminate the algorithm.}
			5- Increment $m$.}
		Return $\hat{\mathbf{x}} \leftarrow \mathbf{x}_{[m]}$ and $\hat{p} \leftarrow p $.		
		\caption{Computation of the Steady-State Distribution Vector $\mathbf{x}$ for RA-UL Scheme.}
		\label{dist_algo_NB}
		\normalsize
	\end{algorithm}
\end{figure}

\subsubsection{RA-UL Scheme}The spatiotemporal iterative algorithm for the SC-UL scheme is given in the following corollary.

\begin{corollary}
\label{corollary1}
The probability  of being idle for the RA-UL scheme is obtained via {\bf Algorithm \ref{dist_algo_NB}}.

\begin{proof}
Similar to Theorem 1.
\end{proof}
\end{corollary}

Using the output of {\bf Algorithm \ref{dist_algo_NB}}, the average queue length  for a stable device in the RA-UL scheme is
 \begin{align} \label{ave_queue_baseline}
\mathbb{E}\left\{\mathcal{Q}_{\text{\rm L}}\right\} = \frac{a^2 (1-\hat{p})\hat{x_\circ}}{(\hat{p}-a)^2}.
\end{align}
{Moreover, the waiting time in the queue is given by:
\begin{align}\label{wait_geo}
\mathbb{P} \left\{\mathcal{W_{\text{q}}}=0 \right\} &=\hat{x_\circ}, \; \text{and} \quad \mathbb{P}\left\{\mathcal{W_{\text{q}}}=j \right\} =\displaystyle{\sum\limits_{v=1}^{j} \hat{x_v} B_j^{(v)}}, \; \text{with} \notag  \\
\quad B^{(k)}_j  &= \left\{\begin{matrix}
(1-\hat{p})^{j-1}\hat{p}  & k=1, j \geq 1  \\
\hat{p}^k  &   j=k, k\geq 1 \\
(1-\hat{p}) \; B^{(k)}_{j-1}+\hat{p} \;B^{(k-1)}_{j-1}  & k \geq j\geq 1 \\
\end{matrix}\right.
\end{align}
with an avergae value of 
\begin{align} 
\mathbb{E}\left\{\mathcal{W_{\text{q}}}\right\} = \frac{ a \bar{a}  \hat{x_\circ}}{(\hat{p}-a)^2},
\end{align}
the variance and the index of dispersion can be computed by substituting \eqref{wait_geo} in \eqref{wait_var} and \eqref{ind_dis}. }

\section{ Numerical Results}\label{sec:Results}

\begin{table}[t]
	\centering
	\renewcommand{\arraystretch}{1.3}
	\caption{Simulation Parameters}
	\begin{tabular}{||p{1.2cm}|p{6cm}|p{3cm}|}
		\hline  \textbf{Notation} & \textbf{Description} & \textbf{Value}\\
		\hline $\alpha$ & devices-to-BS ratio & $(0,500]$  device/BS \\
		\hline $\eta$ & path-loss exponent & $4$ \\
		\hline $a$ & geometric arrival parameter & $.1$ \\
		\hline $\rho$  & power control threshold & $ -90$ dBm  \\
		\hline $\sigma^2$ &  noise power & $-90$ dBm\\			
		\hline $\theta_{\text{\rm SR}} $ & detection threshold for RA-SR  & $-7$ dB\\
		\hline $\theta_{\text{\rm Tx}} $ &  detection threshold for {\rm EA-Tx}  & $-5$ dB\\
		\hline $n_{\text{\rm Z}}$ & number of ZC codes dedicated for RA-SR & $64$ code per BS\\
		\hline $N$ &  number of allocated time slots for {\rm EA-Tx} & $3$ and $6$ \\
		\hline $q$ & number of resource blocks for {\rm EA-Tx}  & $50$ and $105$ \\
		\hline $\theta_{\text{\rm UL}} $ &  detection threshold for RA-UL  & $-5$ dB\\
		\hline $n_{\text{\rm c}}$ & number of resource blocks RA-UL  &  $55$ and $110$ \\
		\hline
	\end{tabular}
	\label{table_parameters}
\end{table}

This section validates the developed spatiotemporal model via independent  Monte Carlo simulation and presents some numerical results to assess and compare the performance of the SC-UL and RA-UL schemes. It is important to note that the simulation is used to verify the stochastic geometry analysis for the transmission success probabilities, i.e., to validate Approximations \ref{app1}-\ref{app3} as well as the approximation of PDF of the Voronoi cell area while calculating the distribution of the number of users in the cell. On the other hand, the queueing analysis is exact, and hence, is embedded into the simulation. {In each simulation run, the BSs and IoT devices are realized over a 100 km$^2$ area via independent PPPs according to the steady-state distribution. Each IoT device is associated to its nearest BS and employs channel inversion power control. The collected statistics are taken for devices located within 1 km from the origin to avoid the edge effects. The received SINR for each device is measured  and a successful transmission is  reported if the SINR is greater than the detection threshold $\theta_{\text{\rm Tx}} $ for  {\rm EA-Tx}. On the other hand, only the device with the highest SINR succeeds if its SINR exceeds the UL SINR threshold $\theta_{\text{\rm SR}} $ and  $\theta_{\text{\rm UL}} $ for {\rm RA-SR} and  {\rm RA-UL} transmissions, respectively.  }Without loss in generality, the system parameters used for this section are reported in Table~\ref{table_parameters}. We consider two operating scenarios for the total available bandwidth, namely, $10$ MHz and $20$ MHz. For the SC-UL scheme, the RA-SR takes place over $1$ MHz and each resource block occupies 180 kHz. As a result, the number of available recourse blocks $(q)$ are $50$ for the $10$ MHz and $105$ for the $20$ MHz. On the other hand, all the available spectrum can be utilized for data communications in the RA-UL scheme which makes the available number of resource blocks $55$ for the $10$ MHz and $110$ for the $20$ MHz, which is 10$\%$ more than the {\rm EA-Tx} resource blocks in the SC-UL scheme.

\begin{figure}
	\begin{center}
		\subfigure[Total Bandwidth $10$ MHz.]{\label{fig:RA1} \includegraphics[width=3.15 in, height=2.6 in]{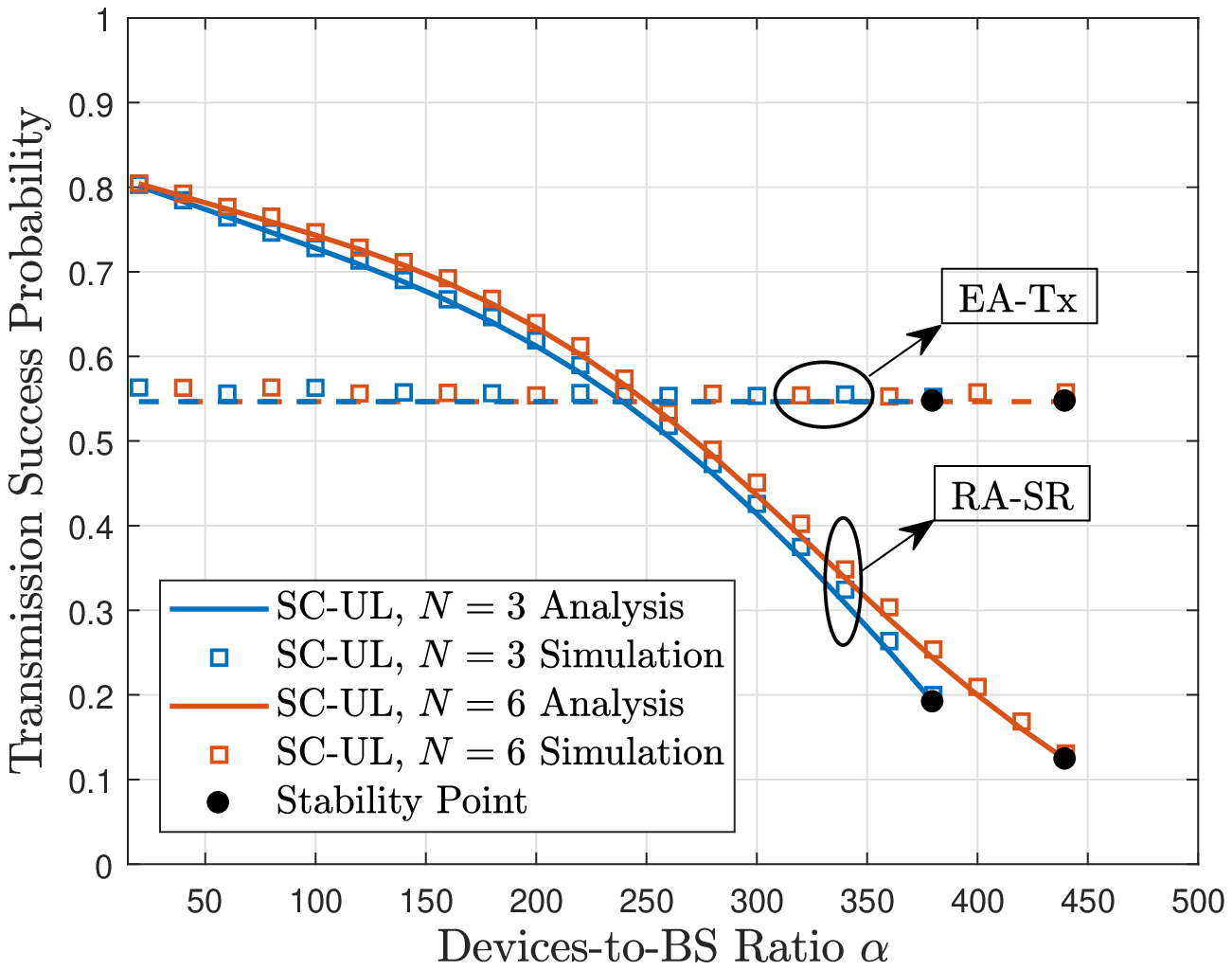}}
		\subfigure[Total Bandwidth $20$ MHz.]{\label{fig:RA2} \includegraphics[width=3.15 in, height=2.6 in]{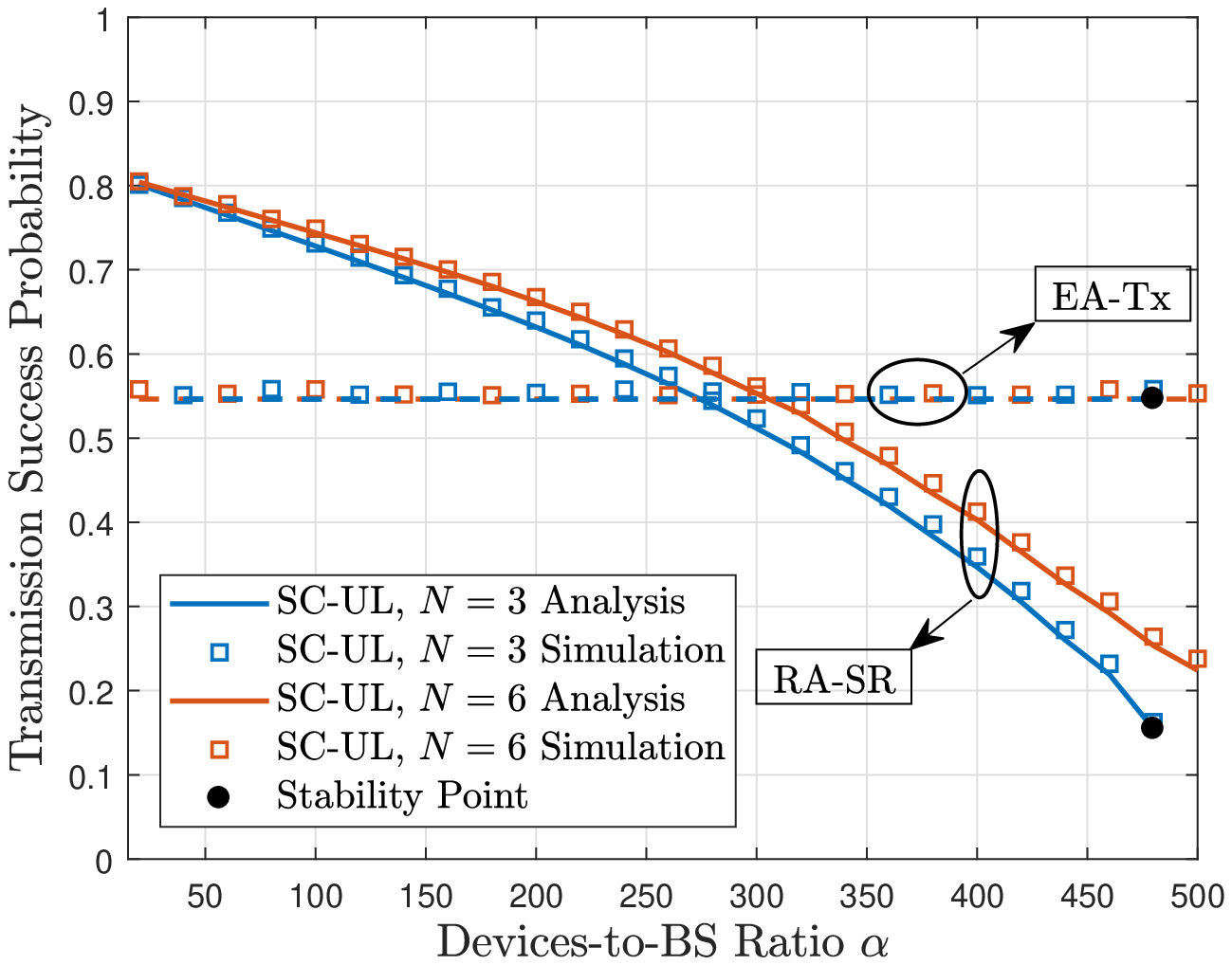}}
	\end{center}
	\vspace{-5mm}
	\caption{ Success probabilities for SC-UL transmission. }
	\label{fig_SINR_RA}
\end{figure}

Fig.~\ref{fig_SINR_RA} shows the transmission success probabilities for the SC-UL scheme at steady-state versus the devices-to-BS ratio $(\alpha)$. It is important to note the close match between the analysis and simulation results which validates the developed mathematical framework. Obviously, by comparing Fig.~\ref{fig:RA1} with Fig.~\ref{fig:RA2}, when the total bandwidth increases the RA-SR success probability increases, this is mainly because of the probability of available resources $P_{\text{\rm aval}}$ increase when there are more UL frequency channels, which in turn reliefs the RA-SR intracell interference by accommodating more devices after RA-SR success. Moreover, Fig.~\ref{fig_SINR_RA} also shows the {\rm EA-Tx} transmission success probabilities. Note that the steady-state value of the {\rm EA-Tx} scheme is less than that of the RA-SR scheme at low device density because $\theta_{\text{\rm Tx}} > \theta_{\text{\rm SR}}$. However, as the device density in the RA-SR increase, the {\rm EA-Tx} success probability outperforms that of the RA-SR scheme despite that fact that  $\theta_{\text{\rm Tx}} > \theta_{\text{\rm SR}}$.  Hence, Fig.~\ref{fig_SINR_RA} shows that {\rm EA-Tx} enforces a constant ${P_{\text{\rm Tx}}}$ despite the value of $\alpha$ by alleviating intracell interference and allowing only one intercell interferer per BS.  Fig.~\ref{fig:RA1} shows that the queue will fall into instability when the devices intensity, or equivalently $\alpha $, goes beyond $380$ because of the limited resource blocks. Note that the results in  Fig.~\ref{fig_SINR_RA} is consistent with eq. \eqref{eq:Out2}. It is important to highlight that the instability point in Fig.~\ref{fig:RA1} is due to the instability of the RA-SR in Fig.~\ref{fig:RA1}. Hence, despite that the {\rm EA-Tx} provision a constant success probability for the scheduled devices, the SC-UL bottleneck is in the SA-RA process.  Hence, Fig.~\ref{fig_SINR_RA} highlights the benefit/drawback of the SC-UL scheme that can provision a certain QoS for scheduled UL transmission upon RA-SR success. 
\begin{figure}[t]
	\begin{center}
		\includegraphics[width=3.2 in]{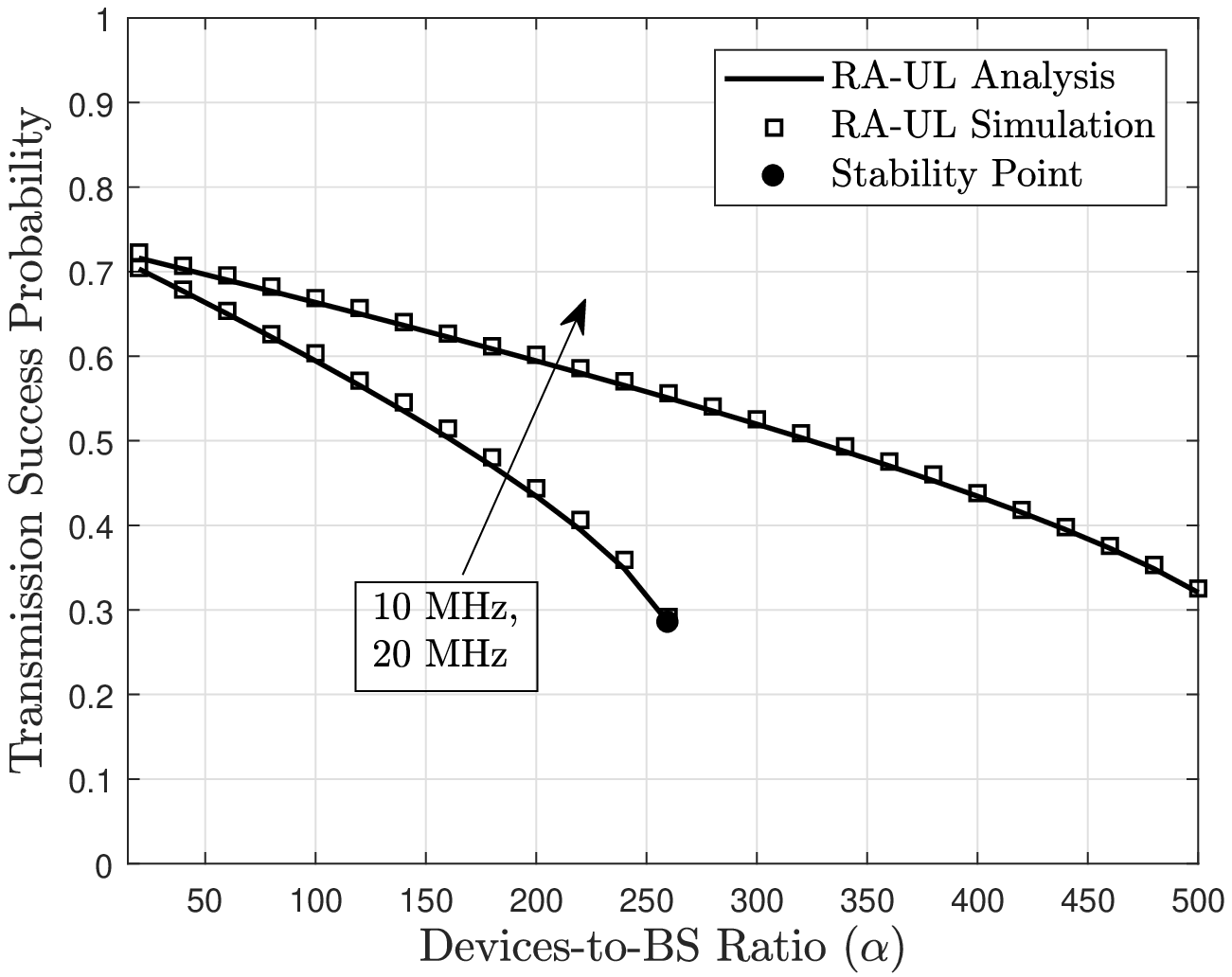}
	\end{center}
	\vspace{-5mm}
	\caption{ Success probability for RA-UL transmission. }\label{fig_SINR_TX_NB}
\end{figure}

%

Fig.~\ref{fig_SINR_TX_NB} shows the RA-UL transmission success probabilities at steady-state versus the devices-to-BS ratio $(\alpha)$. It is important to note the close match between the analysis and simulation results which validates the developed mathematical framework. The figure shows that the performance of the RA-UL transmission is affected by the system load. Hence, the RA-UL scheme cannot provide QoS guarantee for data transmission when compared to the {\rm EA-Tx}. The figure also shows that  the performance of RA-UL can be improved by increasing the number of channels, which diversifies interference and can be used to avoid system instability. {By comparing Fig.~\ref{fig_SINR_RA} and Fig.~\ref{fig_SINR_TX_NB}, the RA-UL shows a better performance than the {\rm EA-Tx} at low device density for $\theta_{\text{\rm Tx}} = \theta_{\text{\rm UL}}$. This is mainly due to the 10$\%$ higher number of resource blocks at the RA-UL scheme, and hence, limited intracel interference at low device density}. However, as the density of the devices increase, the  success probability for the {\rm EA-Tx} scheme outperforms that of RA-UL scheme. It is also worth noting that the success probability for the RA-SR scheme is better than that of the RA-UL scheme because $\theta_{\text{\rm UL}} > \theta_{\text{\rm SR}}$.


\begin{figure}
	\begin{center}
		\subfigure[Total Bandwidth $10$ MHz.]{\label{fig:que1}\includegraphics[width=3.2 in]{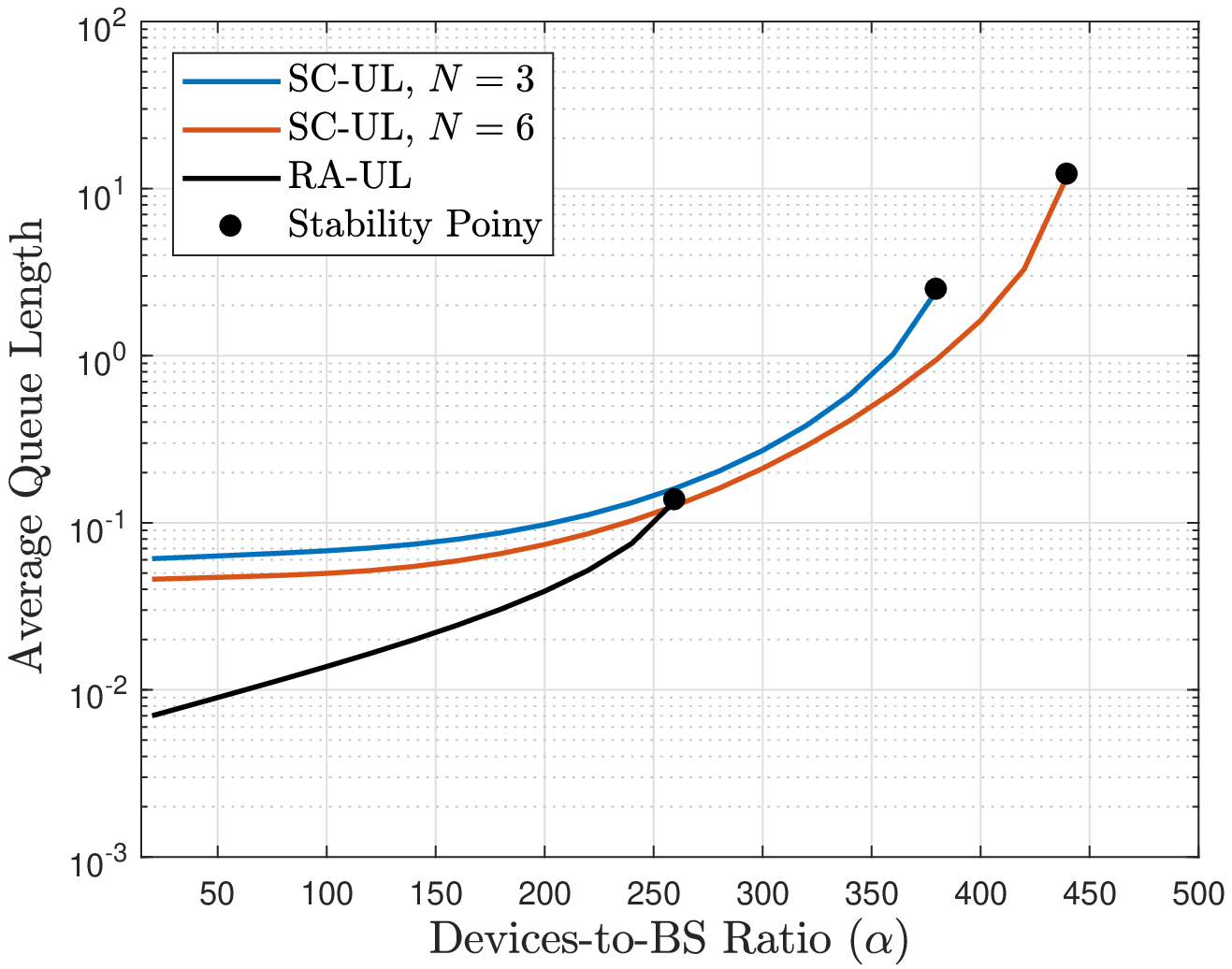}}
		\subfigure[Total Bandwidth $20$ MHz.]{\label{fig:que2}\includegraphics[width=3.2 in]{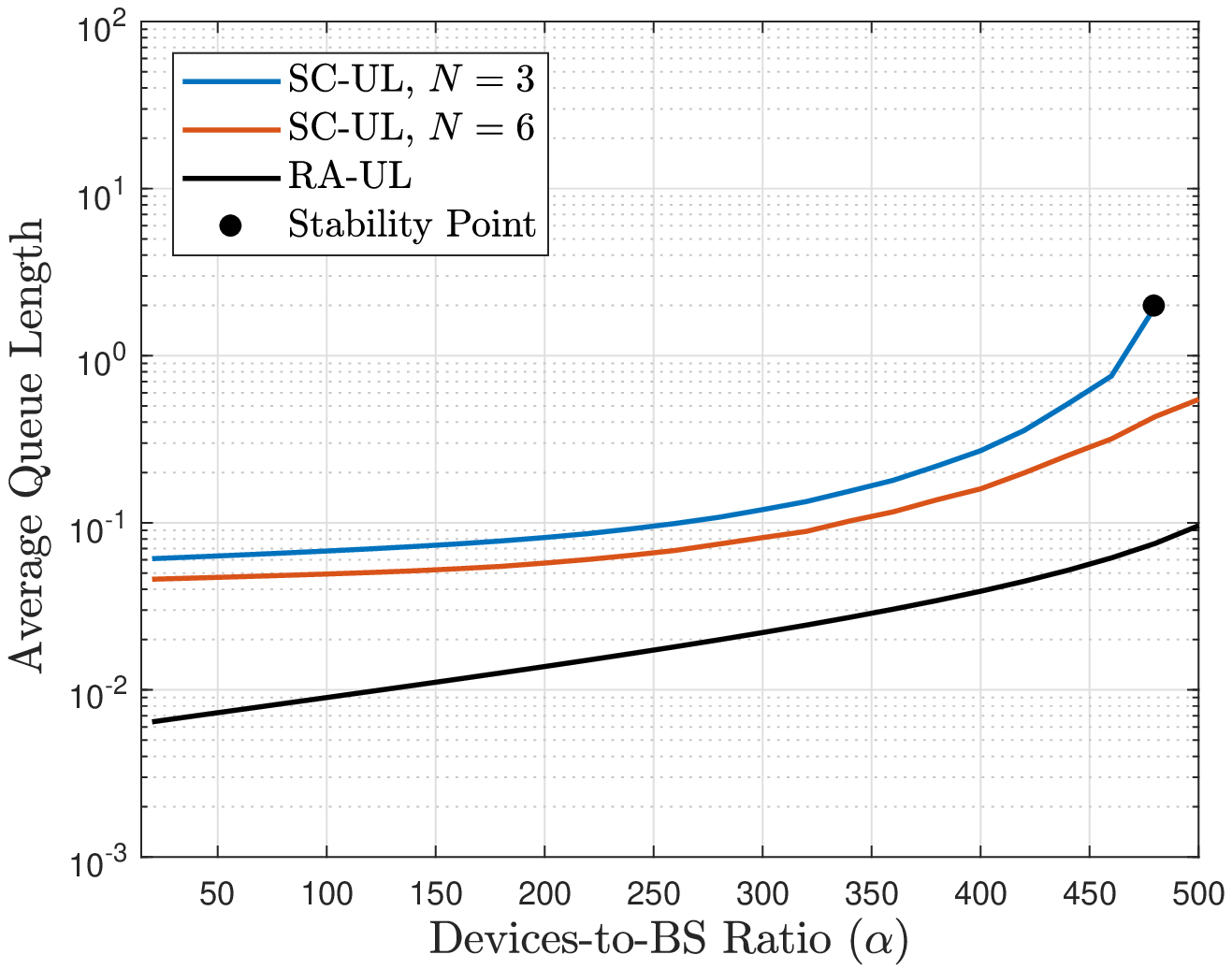}}
	\end{center}
	\vspace{-5mm}
	\caption{ Average queue length for $a=0.1$. }
	\label{fig_queue}
\end{figure}

\begin{figure}
	\begin{center}
		\subfigure[Total Bandwidth $10$ MHz.]{\label{fig:wait1}\includegraphics[width=3.2 in]{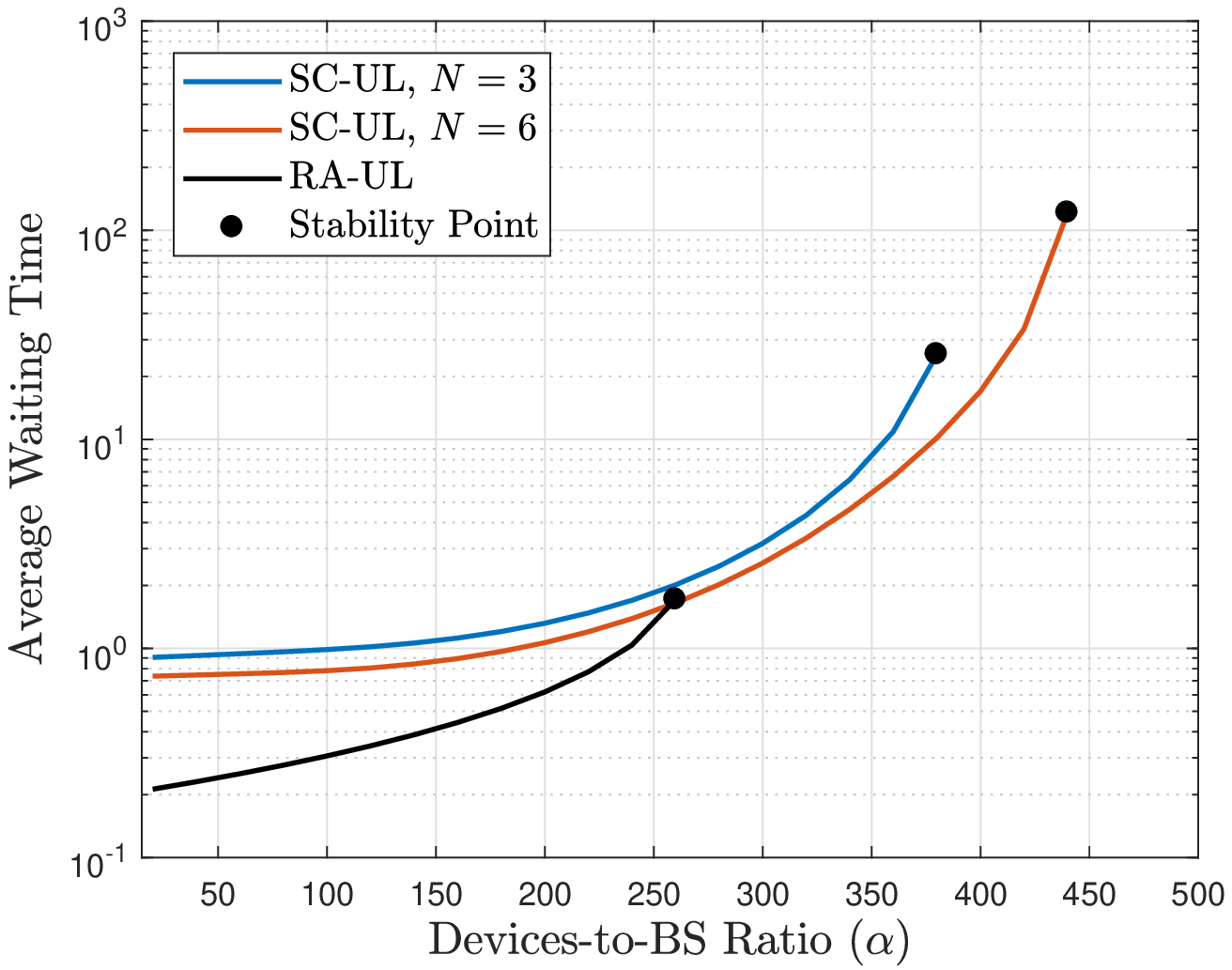}}
		\subfigure[Total Bandwidth $20$ MHz.]{\label{fig:wait2}\includegraphics[width=3.2 in]{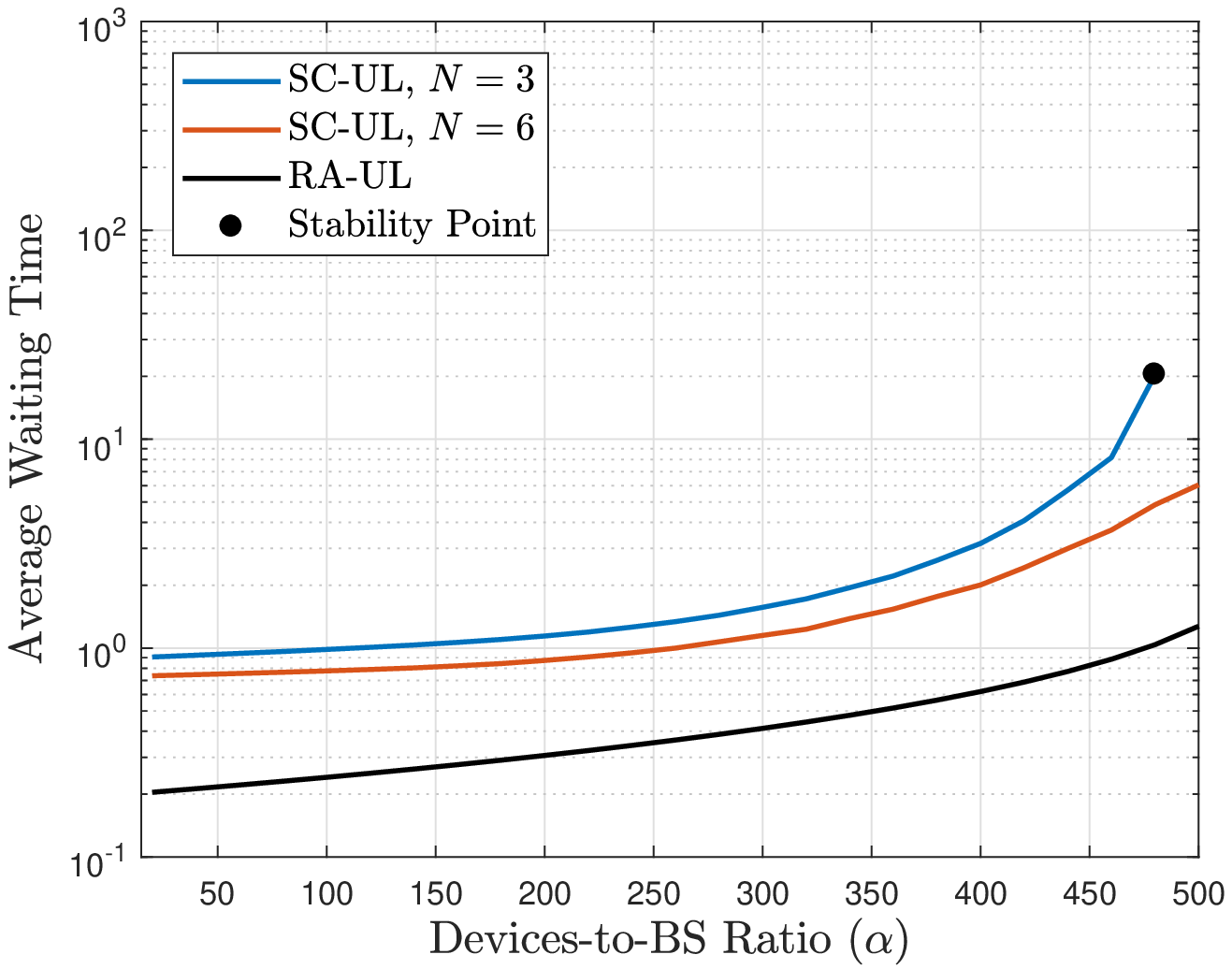}}
	\end{center}
	\vspace{-5mm}
	\caption{ Average queueing waiting time for $a=0.1$.}
	\label{fig_waiting}
\end{figure}

\begin{figure*}
	\begin{center}
		\subfigure[Total Bandwidth $10$ MHz.]{\label{fig:out1}\includegraphics[width=3.2 in , height=2.6 in]{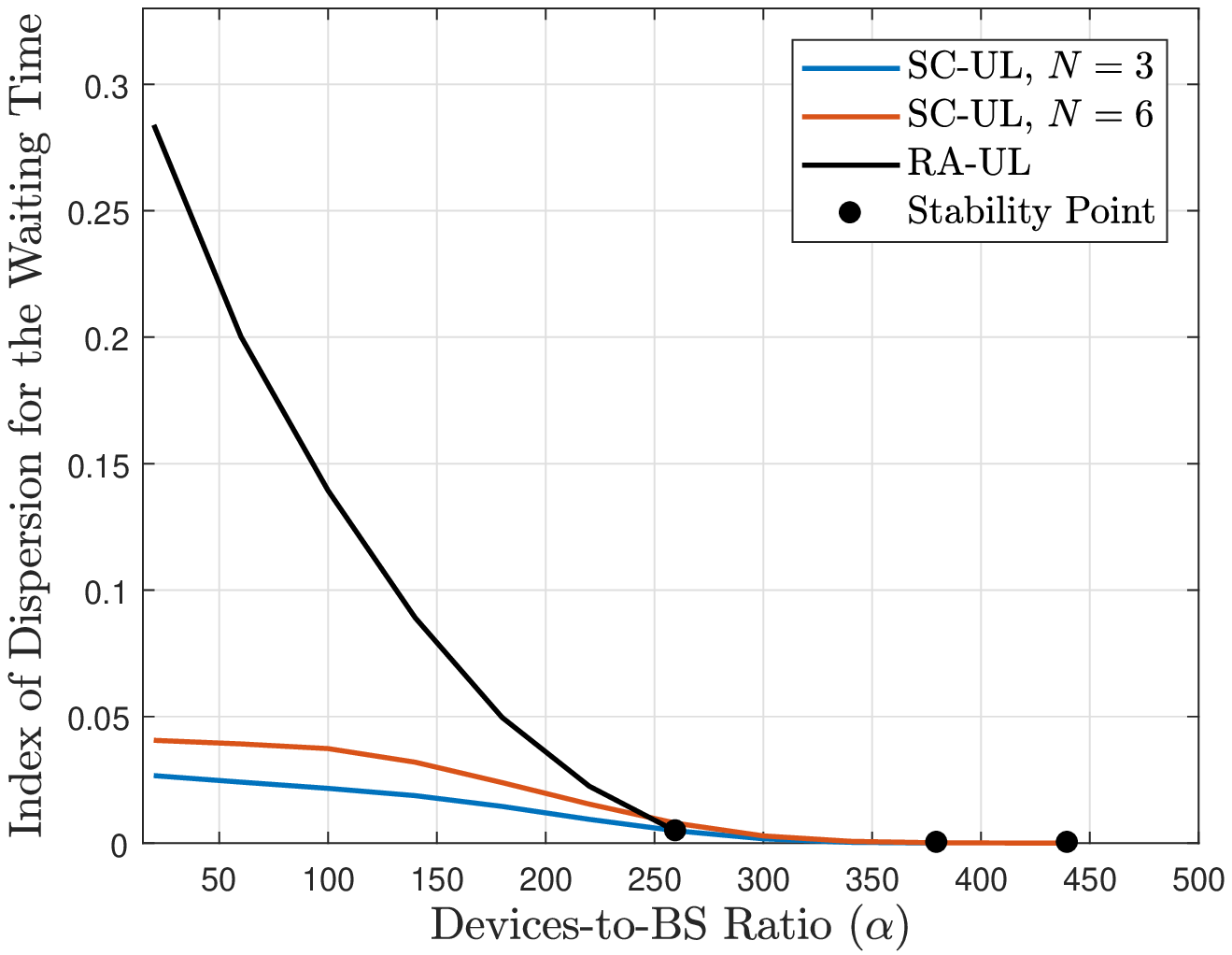}}
		\subfigure[Total Bandwidth $20$ MHz.]{\label{fig:out2}\includegraphics[width=3.2 in, height=2.6 in]{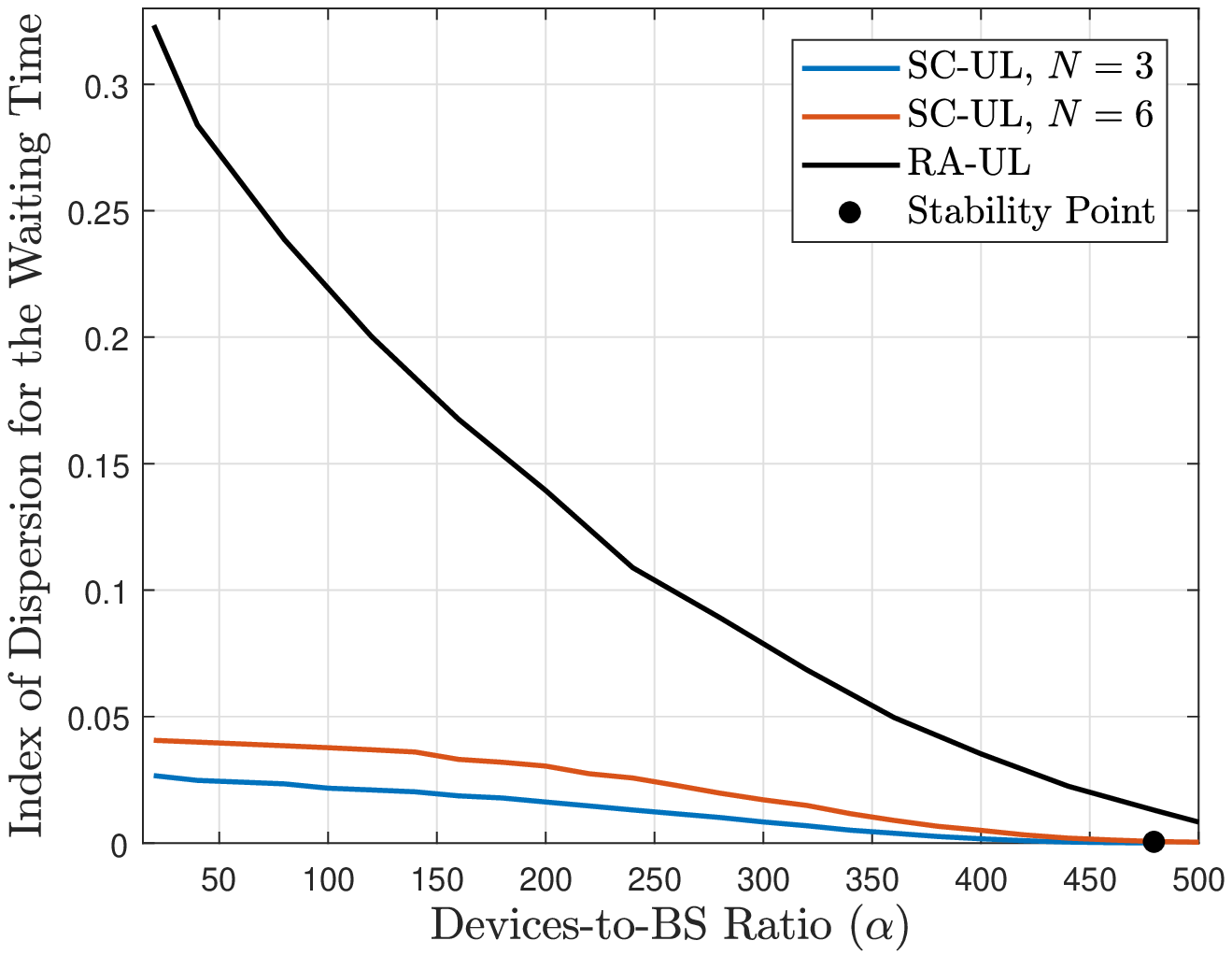}}
	\end{center}
	\vspace{-3mm}
	\caption{ {The index of dispersion for the queueing waiting time for $a=0.1$. }}
	\label{fig_var}
\end{figure*}

\begin{figure*}
	\begin{center}
		\subfigure[Total Bandwidth $10$ MHz.]{\label{fig:out1}\includegraphics[width=3.2 in , height=2.6 in]{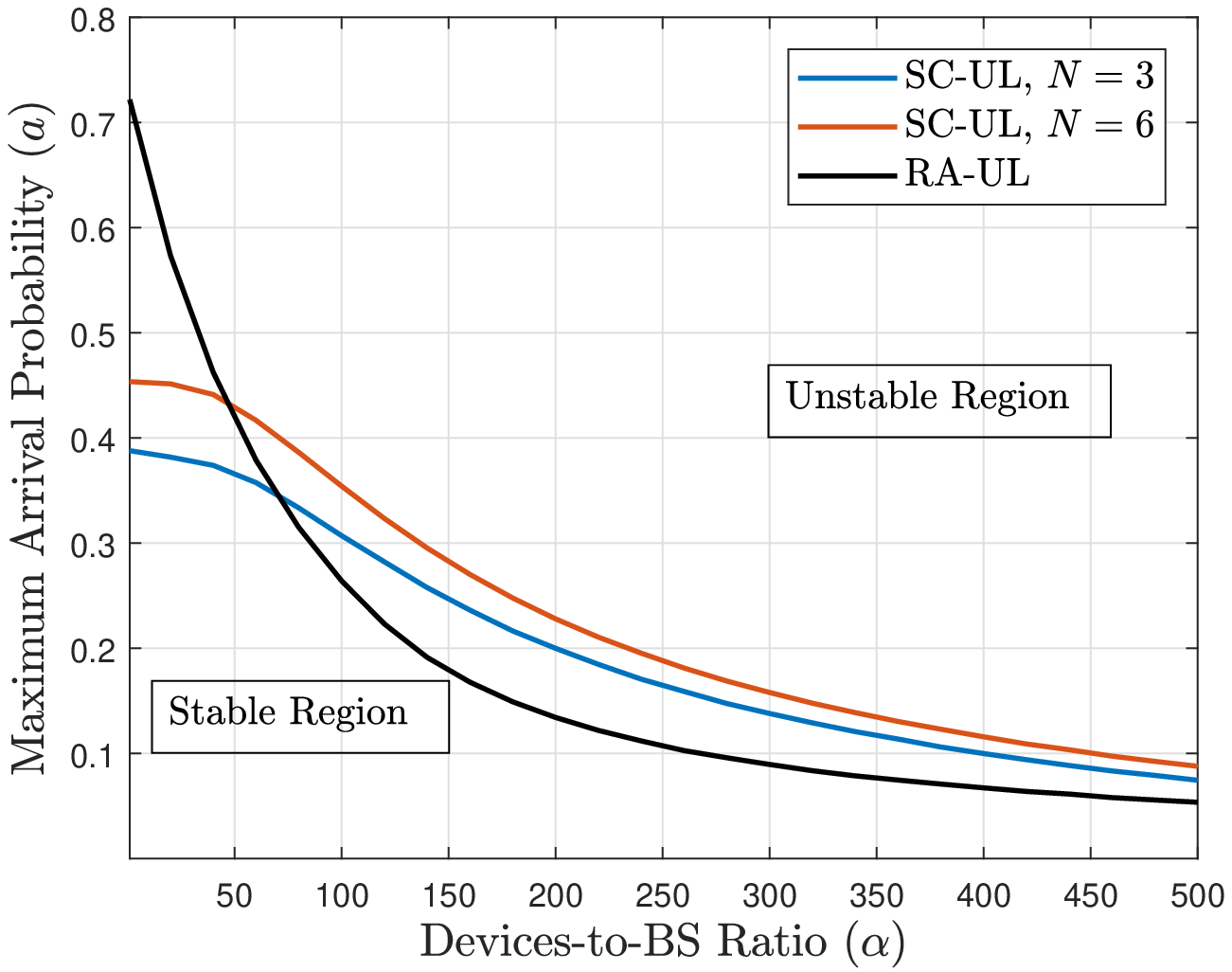}}
		\subfigure[Total Bandwidth $20$ MHz.]{\label{fig:out2}\includegraphics[width=3.2 in, height=2.6 in]{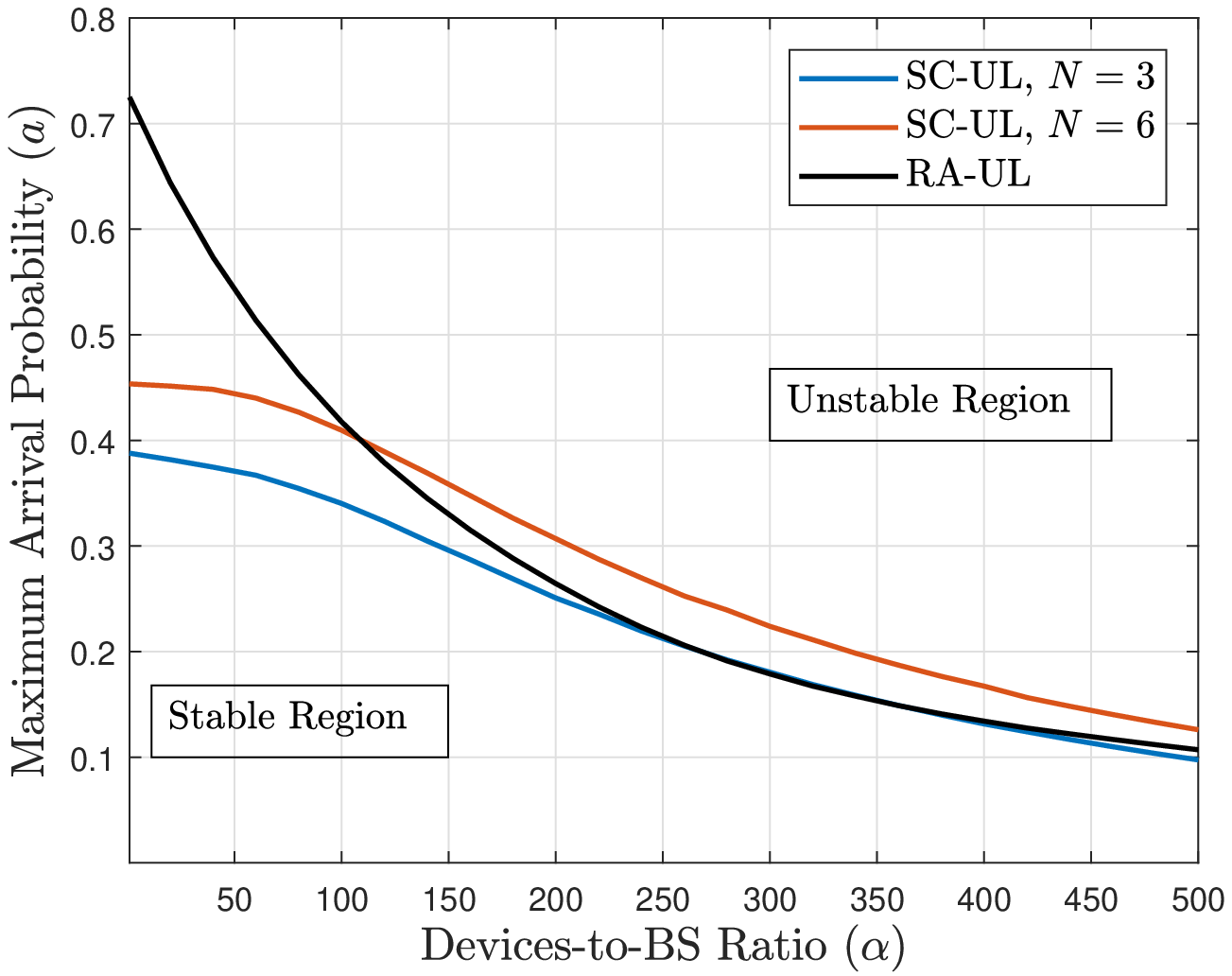}}
	\end{center}
	\vspace{-3mm}
	\caption{  Pareto-frontier of the stability region with respect to the devices intensity. }
	\label{fig_scalability}
\end{figure*}

%


Fig.~\ref{fig_queue} and Fig.~\ref{fig_waiting} show, respectively, the steady-state average queue length $\mathbb{E}\{\mathcal{Q}_L\}$ and  the average waiting time $\mathbb{E}\{\mathcal{W_{\text{q}}}\}$ at stable network operation for the SC-UL and the RA-UL schemes. Comparing both transmission schemes, the figures show that the prompt transmission of the RA-UL scheme offers lower average queue size and delay as long as the network is stable. Hence,  Figs.~\ref{fig_queue}~and~\ref{fig_waiting} support the intuition that prompt transmission of the packets, even without scheduling, expedite packet delivery and helps devices to flush their buffers soon after packets generation.  However, as the devices density increases, the interference becomes overwhelming and scheduling is necessary. Hence, the SC-UL scheme extends the system stability for higher devices density for scarce resources scenario (i.e., the 10 MHz scenario). {Comparing Fig.~\ref{fig_queue} and Fig.~\ref{fig_waiting} also reveals the effect of the {\rm EA-TX} transmission slots ($N$) in the {\rm SC-UL} scheme. The figures show that the higher the $N$ the  {\rm SC-UL} scheme has better average performance in terms of queue length and waiting time. However, this improved performance comes at the expense of a higher index of dispersion for waiting time as depicted in Fig.~\ref{fig_var}. I.e.,  the waiting time for the packets will have higher deviation from the mean $N$ increases.   Fig.~\ref{fig_var} also shows that {\rm RA-UL} scheme generally has higher index of dispersion for the waiting time than the  {\rm SC-UL}. Moreover, the figure shows that the variance decreases as the intensity increases. This is mainly due to the severe interference level at high intensities, and hence, the packets experience significantly large waiting time and low index of dispersion. }


To better compare the scalability of the SC-UL and RA-UL scenarios, the Pareto-frontier of the stability regions for both schemes are shown in Fig.~\ref{fig_scalability}. The stability Pareto-frontier identifies the system parameters that guarantee stable system performance. Operating beyond the stability region lead to unstable queues and unbounded delay. For instance, the instability point in Figs.~\ref{fig_queue}~and~\ref{fig_waiting} occurs at $\alpha = 250$ and $a=0.1$ for the RA-UL in the 10 MHz scenario. This point is located at the stability Pareto-frontiers of the RA-UL scheme in Fig~\ref{fig:out1}. Similarly, Figs.~\ref{fig_queue}~and~\ref{fig_waiting}  show that the SC-UL scheme with $N=3$ and $N=6$ become unstable at devices densities of, respectively,  $\alpha \approx 370$ and $\alpha \approx 440$, at $a=0.1$ in the 10 MHz scenario. Such information can also be extracted from the Pareto-frontiers of the SC-UL scheme for $N=3$ and $N=6$ in Fig~\ref{fig:out1}. Hence, the stability region in Fig~\ref{fig:out1} offers insightful information for the scalability,  and identifies the effective operational scenario, of each transmission  scheme. Having said that, Fig~\ref{fig_scalability} shows that RA-UL offers more scalability in terms of traffic intensity and that SC-UL offers more scalability  in terms of the devices density. Particularly, the RA-UL succeeds to support higher traffic intensity for $\alpha < 50$ [$\alpha < 100$] for the 10 MHz [20 MHz] scenario. In this case, the RA-SR would cause unnecessary delay and it is better to promptly transmit UL data packets without scheduling. When the devices intensity increases, RA-UL would lead to overwhelming interference and scheduling becomes a necessity. Consequently, the SC-UL scheme succeeds to support high devices density that cannot be supported by the RA-UL.

%
%

\section{Conclusions}\label{sec:Conclusions}
One of the key challenges associated with the IoT is tremendous growth in the number of uplink connections. The 3GPP community is seemingly set to continue to pursue a scheduled uplink (SC-UL) transmission paradigm. On the other hand, the LPWA community (e.g., the likes of Sigfox and LoRa) have adopted a random access uplink transmission (RA-UL) paradigm. A legitimate question is: {\em which one of those two paradigms is better in the context of the IoT?} Such a dilemma has been faintly tackled in the literature using dominantly qualitative arguments. This paper, however, provides a concrete framework for tackling the question in hand.
The paper develops an integrated model featuring the use of stochastic geometry and queueing theory for uplink transmissions. We  capture the mutual interference between the IoT devices by utilizing spatially interacting tandem queues model. The performance of both schemes is analyzed in terms of 4 key parameters: transmission success probability, buffer queue length, access delay time, and scalability. The latter is expressed in terms of the IoT device intensity and traffic arrival rates.
The key takeaway is that RA-UL is the best choice for lower device intensities and traffic volumes. Beyond that, SC-UL features a more robust performance. The rationale for that stems from the failure of RA-UL to handle the escalation in intracell interference with growing device intensities and traffic rates. Interestingly, this means that LPWA technologies today are actually being operated in their own ``sweet spot'', i.e., where the number of IoT devices per base station is relatively lower than what LTE base stations are engineered for. As such, RA-UL and SC-UL ought to be perceived as complementary paradigms rather than contentious.

\appendix
\vspace{-5mm}
\subsection{Proof of Lemma 1}\label{sec:AppA}
{Let the number of neighbors in the cell equals to $\mathtt{n}$, then the Complementary Cumulative Distribution Function (CCDF) of the maximum channel gain between $\mathtt{n}$ independent Rayleigh fading channel gains is given by:
	\vspace{-2mm}
\begin{align}\label{eq:AppA_12}
\bar{F}_{h_{max}} \mid_{\mathcal{N}=\mathtt{n}}\left(h\right)&= 1-(1-\exp(-h))^{\mathtt{n}+1}.
\end{align}
Substituting \eqref{eq:AppA_12} in \eqref{SINR_RA} and noting that $\mathbb{P}\left\{h_\circ> h_i \;  \forall h_i \in \bold{h}_o \setminus h_\circ \right\}=1/(\mathtt{n}+1)$ yield to:
\begin{align}\label{eq:AppA_13}
P_{\text{\rm RA}} \mid_{\mathcal{N}=\mathtt{n}}&= \frac{\mathbb{E}_{\mathcal{I}^{(\text{\rm RA})}_{\text{\rm Inter}},\mathcal{I}^{(\text{\rm RA})}_{\text{\rm Intra}}}\left\{1-\left( 1-\exp\left\{ \frac{\theta_{ \text{\rm RA}}}{\rho} \left( \sigma^2+\mathcal{I}^{(\text{\rm RA})}_{\text{\rm Inter}}+\mathcal{I}^{(\text{\rm RA})}_{\text{\rm Intra}} \right) \right\} \right)^{\mathtt{n}+1} \right\} }{\mathtt{n}+1}.
\end{align}
 Because of the independency of the PPP in different regions \cite{ martin_book} and after applying the binomial expansion for the numerator of \eqref{eq:AppA_13}, we get:
 \begin{align}\label{eq:AppA_15}
 P_{\text{\rm RA}} \mid_{\mathcal{N}=\mathtt{n}}&= \frac{\sum\limits_{k=1 }^{\mathtt{n}+1} \binom{\mathtt{n}+1}{k}(-1)^{k+1} \exp \left\{\frac{-k \; \theta_{\text{\rm RA}\; }\sigma^2}{\rho}\right\}\mathscr{L}_{\mathcal{I}^{(\text{\rm RA})}_{\text{\rm Inter}}}\left(\frac{k \; \theta_{ \text{\rm RA}}}{\rho} \right) \mathscr{L}_{\mathcal{I}^{(\text{\rm RA})}_{\text{\rm Intra}}\mid_{\mathcal{N}=\mathtt{n}}} \left(\frac{k \; \theta_{ \text{\rm RA}}}{\rho} \right)}{\mathtt{n}+1}.
 \end{align}}


%
	\vspace{-5mm}
Note that the nearest BS association and the employed power control enforce the following two conditions; (i) the intracell interference from an interfering device is equal to $\rho$, and (ii) the intercell interference from any interfering device is strictly less that $\rho$. The aggregate intercell interference received at the serving BS of the test device is obtained as:
\begin{align}\label{eq:AppA_3}
\mathcal{I}^{(\text{\rm RA})}_{\text{\rm Inter}}=  \sum\limits_{u_i\in \Phi \setminus \{o\} } \mathbbm{1}_{\{P_{i} \parallel u_i\parallel^{-\eta}<\rho\}}P_{i} \; h_i \parallel u_i\parallel^{-\eta}.
\end{align}
Ignoring the correlations between the transmission powers of the devices in the same and adjacent Voronoi cells, the LT of \eqref{eq:AppA_3} can be approximated as:
\begin{align}
&\mathscr{L}_{\mathcal{I}^{(\text{\rm RA})}_{\text{\rm Inter}}}\left(s\right) \approx \exp\left(-2\pi \; \varphi_1   \; \tilde{\mu} \; s^{\frac{2}{\eta}} \;\mathbb{E}_{P}\left\{P^{\frac{2}{\eta}} \;   \right\} \int_{(s\rho)^{\frac{-1}{\eta}}}^{\infty} \frac{y}{y^\eta +1} dy  \right).
\end{align}

\noindent  The LT is obtained by using the probability generating function (PGFL) of the PPP \cite{ martin_book} and following \cite{elsawy2014stochastic}, substituting the value of $\mathbb{E}_{P}\left\{P^{\frac{2}{\eta}}\right\}$ from [Lemma 1,\cite{elsawy2014stochastic}] and evaluating the integral gives:
\begin{align} \label{Eq:laplase_Inter_app}
&\mathscr{L}_{\mathcal{I}^{(\text{\rm RA})}_{\text{\rm Inter}}}\left(\frac{k \; \theta_{ \text{\rm RA}}}{\rho}\right) \approx \exp\left(-2 \; k \; \theta_{ \text{\rm RA}}\; \frac{ \varphi_1  \;\tilde{\mu}}{\lambda} \; \frac{{}_2F_1\left(1,1-\frac{2}{\eta},2-\frac{2}{\eta},-k\;\theta_{\text{\rm SR}}\right)}{\eta-2} \right).
\end{align}

{Let the signal from all the devices in the test cell to be $ \rho \bold{h}_\circ=\{\rho {h}_\circ ,\rho h_1, \rho h_2, \hdots, \rho h_{\mathtt{n}}\}$, and let $y$ be the maximum signal among $\mathtt{n}+1$ devices in the test cell, hence to calculate the Intracell interference conditioned on the number of neighbors $\left(\mathscr{L}_{\mathcal{I}^{(\text{\rm RA})}_{\text{\rm Intra}}\mid_{\mathcal{N}=\mathtt{n}}}\right)$, we first find the LT of the truncated exponential PDF as follows:
\begin{align}\label{Laplace_Intra}
\mathscr{L}_{\mathcal{I}^{(\text{\rm RA})}_{\text{\rm Intra}}\mid_{\mathcal{N}=\mathtt{n}}}\left(s\right)= \int_{0}^{y} \frac{\frac{1}{\rho} \exp(-\frac{x}{\rho})}{1-\exp(-\frac{y}{\rho})} \exp(-sx ) \; dx=\frac{\exp(\frac{y}{\rho})-\exp(-sy)}{(1+s\rho)(\exp(\frac{y}{\rho})-1)}.
\end{align}
By deconditioning on the PDF of $y$ which has the form of $(\mathtt{n}+1)(1-\exp(-\frac{y}{\rho}))\exp(-\frac{y}{\rho})$ we get:
\begin{align}
\mathscr{L}_{\mathcal{I}^{(\text{\rm RA})}_{\text{\rm Intra}}\mid_{\mathcal{N}=\mathtt{n}}}\left(\frac{k \; \theta_{ \text{\rm RA}}}{\rho}\right)&=\frac{\mathtt{n}+1}{1+k \; \theta_{\text{\rm SR}}} \left(\frac{1}{\mathtt{n}} -\frac{\Gamma(\mathtt{n}) \;\Gamma(2+k \; \theta_{\text{\rm SR}})}{\Gamma(2+N+k \; \theta_{\text{\rm SR}})} \right).
\end{align}
After applying the law of total probability, \eqref{eq:RA_success} in Lemma \ref{lemma_RA_success} is obtained.}

\bibliographystyle{IEEEtran}
\bibliography{refrences,IEEEabrv}
\vfill
\end{document}